\documentclass[pra,aps,nopacs,onecolumn,twoside,superscriptaddress]{revtex4}

\usepackage{multirow}
\usepackage{amsmath,amsfonts,amssymb,caption,color,epsfig,graphics,graphicx,hyperref,latexsym,mathrsfs,revsymb,theorem,url,verbatim,epstopdf,enumerate}
\usepackage{amsmath,amsfonts,amssymb,caption,hyperref,color,epsfig,graphics,graphicx,latexsym,mathrsfs,revsymb,theorem,url,verbatim,epstopdf,cleveref}
\usepackage{fontenc}
\usepackage{cases}
\hypersetup{colorlinks,linkcolor={blue},citecolor={blue},urlcolor={red}}
\usepackage[T1]{fontenc}
\newtheorem{definition}{Definition}
\newtheorem{proposition}[definition]{Proposition}
\newtheorem{lemma}[definition]{Lemma}

\newtheorem{theorem}[definition]{Theorem}
\newtheorem{corollary}[definition]{Corollary}
\newtheorem{conjecture}[definition]{Conjecture}

\newtheorem{remark}[definition]{Remark}
\newtheorem{example}[definition]{Example}
\newtheorem{question}[definition]{Question}
\newtheorem{memo}[definition]{Memo}

\def\squareforqed{\hbox{\rlap{$\sqcap$}$\sqcup$}}
\def\qed{\ifmmode\squareforqed\else{\unskip\nobreak\hfil
		\penalty50\hskip1em\null\nobreak\hfil\squareforqed
		\parfillskip=0pt\finalhyphendemerits=0\endgraf}\fi}
\def\endenv{\ifmmode\;\else{\unskip\nobreak\hfil
		\penalty50\hskip1em\null\nobreak\hfil\;
		\parfillskip=0pt\finalhyphendemerits=0\endgraf}\fi}
\newenvironment{proof}{\noindent \textbf{{Proof.~} }}{\qed}
\def\Dbar{\leavevmode\lower.6ex\hbox to 0pt
	{\hskip-.23ex\accent"16\hss}D}
\makeatletter
\def\url@leostyle{%
	\@ifundefined{selectfont}{\def\UrlFont{\sf}}{\def\UrlFont{\small\ttfamily}}}
\makeatother
\def\bcj{\begin{conjecture}}
	\def\ecj{\end{conjecture}}
\def\bcr{\begin{corollary}}
	\def\ecr{\end{corollary}}
\def\bd{\begin{definition}}
	\def\ed{\end{definition}}
\def\bea{\begin{eqnarray}}
	\def\eea{\end{eqnarray}}
\def\bem{\begin{enumerate}}
	\def\eem{\end{enumerate}}
\def\bex{\begin{example}}
	\def\eex{\end{example}}
\def\bim{\begin{itemize}}
	\def\eim{\end{itemize}}
\def\bl{\begin{lemma}}
	\def\el{\end{lemma}}
\def\bma{\begin{bmatrix}}
	\def\ema{\end{bmatrix}}
\def\bpf{\begin{proof}}
	\def\epf{\end{proof}}
\def\bpp{\begin{proposition}}
	\def\epp{\end{proposition}}
\def\bqu{\begin{question}}
	\def\equ{\end{question}}
\def\br{\begin{remark}}
	\def\er{\end{remark}}
\def\bt{\begin{theorem}}
	\def\et{\end{theorem}}
\def\bmm{\begin{memo}}
	\def\emm{\end{memo}}

\def\btb{\begin{tabular}}
	\def\etb{\end{tabular}}

	\newcommand{\nc}{\newcommand}
	
	\def\a{\alpha}
	\def\b{\beta}
	\def\g{\gamma}

	\def\z{\zeta}
	
	\def\t{\theta}

	\def\l{\lambda}

	\def\x{\xi}
	
	\def\r{\rho}

	\nc{\bbA}{\mathbb{A}} \nc{\bbB}{\mathbb{B}} \nc{\bbC}{\mathbb{C}}
	\nc{\bbD}{\mathbb{D}} \nc{\bbE}{\mathbb{E}} \nc{\bbF}{\mathbb{F}}
	\nc{\bbG}{\mathbb{G}} \nc{\bbH}{\mathbb{H}} \nc{\bbI}{\mathbb{I}}
	\nc{\bbJ}{\mathbb{J}} \nc{\bbK}{\mathbb{K}} \nc{\bbL}{\mathbb{L}}
	\nc{\bbM}{\mathbb{M}} \nc{\bbN}{\mathbb{N}} \nc{\bbO}{\mathbb{O}}
	\nc{\bbP}{\mathbb{P}} \nc{\bbQ}{\mathbb{Q}} \nc{\bbR}{\mathbb{R}}
	\nc{\bbS}{\mathbb{S}} \nc{\bbT}{\mathbb{T}} \nc{\bbU}{\mathbb{U}}
	\nc{\bbV}{\mathbb{V}} \nc{\bbW}{\mathbb{W}} \nc{\bbX}{\mathbb{X}}
	\nc{\bbZ}{\mathbb{Z}}
	
	\nc{\bA}{{\bf A}} \nc{\bB}{{\bf B}} \nc{\bC}{{\bf C}}
	\nc{\bD}{{\bf D}} \nc{\bE}{{\bf E}} \nc{\bF}{{\bf F}}
	\nc{\bG}{{\bf G}} \nc{\bH}{{\bf H}} \nc{\bI}{{\bf I}}
	\nc{\bJ}{{\bf J}} \nc{\bK}{{\bf K}} \nc{\bL}{{\bf L}}
	\nc{\bM}{{\bf M}} \nc{\bN}{{\bf N}} \nc{\bO}{{\bf O}}
	\nc{\bP}{{\bf P}} \nc{\bQ}{{\bf Q}} \nc{\bR}{{\bf R}}
	\nc{\bS}{{\bf S}} \nc{\bT}{{\bf T}} \nc{\bU}{{\bf U}}
	\nc{\bV}{{\bf V}} \nc{\bW}{{\bf W}} \nc{\bX}{{\bf X}}
	\nc{\bZ}{{\bf Z}}
	
	\nc{\as}{{\cal AS}}
	\nc{\app}{{\cal AP}}
	\nc{\ar}{{\cal AR}}

	\nc{\cA}{{\cal A}} \nc{\cB}{{\cal B}} \nc{\cC}{{\cal C}}
	\nc{\cD}{{\cal D}} \nc{\cE}{{\cal E}} \nc{\cF}{{\cal F}}
	\nc{\cG}{{\cal G}} \nc{\cH}{{\cal H}} \nc{\cI}{{\cal I}}
	\nc{\cJ}{{\cal J}} \nc{\cK}{{\cal K}} \nc{\cL}{{\cal L}}
	\nc{\cM}{{\cal M}} \nc{\cN}{{\cal N}} \nc{\cO}{{\cal O}}
	\nc{\cP}{{\cal P}} \nc{\cQ}{{\cal Q}} \nc{\cR}{{\cal R}}
	\nc{\cS}{{\cal S}} \nc{\cT}{{\cal T}} \nc{\cU}{{\cal U}}
	\nc{\cV}{{\cal V}} \nc{\cW}{{\cal W}} \nc{\cX}{{\cal X}}
	\nc{\cZ}{{\cal Z}}
	
	\nc{\cpp}{{\cal PP}}
	
	\nc{\hA}{{\hat{A}}} \nc{\hB}{{\hat{B}}} \nc{\hC}{{\hat{C}}}
	\nc{\hD}{{\hat{D}}} \nc{\hE}{{\hat{E}}} \nc{\hF}{{\hat{F}}}
	\nc{\hG}{{\hat{G}}} \nc{\hH}{{\hat{H}}} \nc{\hI}{{\hat{I}}}
	\nc{\hJ}{{\hat{J}}} \nc{\hK}{{\hat{K}}} \nc{\hL}{{\hat{L}}}
	\nc{\hM}{{\hat{M}}} \nc{\hN}{{\hat{N}}} \nc{\hO}{{\hat{O}}}
	\nc{\hP}{{\hat{P}}} \nc{\hR}{{\hat{R}}} \nc{\hS}{{\hat{S}}}
	\nc{\hT}{{\hat{T}}} \nc{\hU}{{\hat{U}}} \nc{\hV}{{\hat{V}}}
	\nc{\hW}{{\hat{W}}} \nc{\hX}{{\hat{X}}} \nc{\hZ}{{\hat{Z}}}
	
	\nc{\hn}{{\hat{n}}}
	
	\def\diag{\mathop{\rm diag}}

	\def\max{\mathop{\rm max}}
	\def\min{\mathop{\rm min}}

	\def\tr{\mathop{\rm Tr}}

	\def\dg{\dagger}

	\newcommand{\ket}[1]{|#1\rangle}
	\newcommand{\proj}[1]{| #1\rangle\!\langle #1 |}

	\def\Dbar{\leavevmode\lower.6ex\hbox to 0pt
		{\hskip-.23ex\accent"16\hss}D}

\begin{document}
	
\title{On the extreme points of sets of absolulely separable and PPT states}

\author{Zhiwei Song}\email[]{zhiweisong@buaa.edu.cn}
\affiliation{School of Mathematics and Systems Science, Beihang University, Beijing 100191, China}
 
\author{Lin Chen}\email[]{linchen@buaa.edu.cn (corresponding author)}
\affiliation{School of Mathematics and Systems Science, Beihang University, Beijing 100191, China}

\begin{abstract}
The absolutely separable (resp. PPT) states remain separable (resp. positive partial transpose) under any global unitary operation. We present a compact form of the extreme points in the sets of absolutely separable states and PPT states in two-qubit and qubit-qudit systems. The results imply that each extreme point has at most three distinct eigenvalues. We establish a necessary and sufficient condition for determining extreme points of the set of absolutely PPT states in two-qutrit and qutrit-qudit systems, expressed as solvable linear equations. We also demonstrate that any extreme point in qutrit-qudit system has at most seven distinct eigenvalues. We introduce the concept of robustness of nonabsolute separability. It quantifies the minimal amount by which a state needs to mix with other states such that the overall state is absolutely separable. We show that the robustness satisfies positivity, invariance under unitary transformation, monotonicity and convexity, so it is a good measure within the resource theory of nonabsolute separability. Analytical expressions for this measure are given for pure states in arbitrary system and rank-two mixed states in two-qubit system.
\end{abstract}

\date{ \today }

\maketitle


\section{Introduction}
\label{sec:int}

Quantum entanglement is one of the fundamental features in the quantum information theory \cite{horodecki2009quantum,nielsen2010quantum}. 
A quantum state $\r\in M_m(\bbC)\otimes M_n(\bbC)$
is called separable if it can be written as $\r=\sum_i p_i\proj{v_i}\otimes \proj{w_i}$, with $p_i\ge 0$, $\sum_i p_i=1$, $\ket{v_i}\in\bbC^m$ and $\ket{w_i}\in\bbC^n$. Otherwise it is called entangled, and can  
be considered as a valuable resource for a variety of information processing tasks \cite{chitambar2019quantum,1993Teleporting,gisin2002quantum}. Despite extensive efforts in entanglement detection (readers may refer to \cite{guhne2009entanglement} for a comprehensive review),  determining the separability of a quantum state remains challenging, as it has been proven to be an NP-hard problem \cite{gurvits2003classical}. Nevertheless, there are
some partial results that prove separability of
certain states.
The famous Peres criterion indicates that any separable state $\r$ has positive partial transpose (PPT), i.e., $(\text{id}_m\otimes T)(\r)$ remains positive semidefinite, where $\text{id}_m$ is the identity map on $M_m(\bbC)$ and 
$T$ is the transpose map on $M_n(\bbC)$ \cite{peres1996separability}. In particular,  the PPT states are separable in two-qubit and qubit–qutrit systems \cite{horodecki1996necessary}. However, in higher dimensions, there exist PPT entangled states, indicating that separable states constitute a proper subset of PPT states \cite{horodecki1997separability}. From a geometric perspective, both the sets of separable states and PPT states are characterized as convex and compact. According to the Krein-Milman theorem, a compact convex set is represented as the convex hull of its extreme points \cite{krein1940extreme}. Hence, a crucial approach to understanding these two sets involves identifying their extreme points. For the set of separable states, it is known that the extreme points correspond to the pure product states. Properties about the boundary of this set have also been investigated \cite{chen2015boundary}.
In the context of the set of PPT states, a necessary and sufficient condition for determining its extreme points has been outlined \cite{leinaas2007extreme}, followed by the application of numerical methods to obtain extreme PPT states of varying ranks \cite{leinaas2010numerical}.

An interesting problem related to separability is to study absolutely separable (AS) states, which are states that remain separable under any global unitary transformation \cite{knill2003separability}. Absolute separability is a spectral property, and the problem is to find conditions on the spectrum of a state for it to be AS. One motivation for this problem is that it is experimentally easier to determine the eigenvalues of a state rather than reconstructing the state itself \cite{ekert2002direct,tanaka2014determining}.
From the perspective of resource theory, non-absolutely separable states, which comprise separable and entangled states, can be viewed as a resource, with the absolutely separable states serving as the free states and a mixture of global unitary operations as the free operations \cite{2023Resource}.
Analogous to AS states, states that remain PPT under any global unitary transformation are termed absolutely PPT (AP) states. In the ensuing discussion, we will denote the sets of AS states and AP states in 
$M_m(\bbC)\otimes M_n(\bbC)$ as $\as_{m,n}$ and $\app_{m,n}$, respectively. It directly follows that $\as_{m,n}\subseteq \app_{m,n}$. The characterization of $\as_{2,2}$ was initially provided in \cite{verstraete2001maximally}.
 A necessary and sufficient condition for a state to belong to $\app_{m,n}$ was established, represented by a finite set of linear matrix inequalities  \cite{hildebrand2007positive}. Later, it was proved that $\as_{2,n}=\app_{2,n}$ for arbitrary $n$ \cite{johnston2013separability}. However, the problem of whether the two sets $\as_{m,n}$ and $\app_{m,n}$ are identical for $m,n\ge 3 $ still remains open.

Geometrically, the two sets $\as_{m,n}$ and $\app_{m,n}$ are also known to be convex and compact \cite{ganguly2014witness}. Therefore, a potential approach to this question involves determining whether every extreme point of $\app_{m,n}$ belongs to $\as_{m,n}$. Researches have concentrated on $\as_{2,n}$, showing that every deficient-rank AS state is an extreme point \cite{halder2021characterizing}. On another front, various geometrical measures of nonabsolute separability (NAS) for a state have been introduced \cite{2023Resource}. These measures, assessing the distance from a state to the set of AS states, are more fine-grained than the entanglement measures, as they can detect certain separable states.

The goal of the paper is to characterize more geometric properties of $\as_{m,n}$ and $\app_{m,n}$, focusing on their boundary points and extreme points. In Section \ref{sec:pre}, we introduce
the mathematical methods that we will use, and investigate some results about AS and AP states. We present a fundamental property regarding the non-extreme points of $\as_{m,n}$ and  $\app_{m,n}$ in Theorem \ref{thm:main}. In Section \ref{ex2n}, we give a full characterization of extreme points of $\as_{2,n}$ and $\app_{2,n}$ for arbitrary $n$ in Theorem \ref{mmain}. The results imply that each extreme point has at most three distinct eigenvalues. In Section \ref{sec:3x3}, we establish a sufficient and necessary condition on determining the extreme points of $\app_{3,3}$ in Theorem \ref{judge}. This condition can be expressed as a set of solvable linear equations. As a corollary, we present a family of extreme points in $\app_{3,3}$ which have two distinct eigenvalues. Furthermore, we extend the necessary and sufficient condition for determining the extreme points of $\app_{3,n}$ in Theorem \ref{3ndpd}. For convenience of the readers, we summarize the results in Table
\ref{tab}. In section \ref{ase}, we introduce the concept of robustness of nonabsolute separability, which is proven to be a new distance-based measure of NAS. In Theorem \ref{battle}, we provide an exact formula of certain states, including pure states as well as arbitrary rank-two two-qubit mixed states, under this measure. Finally, we conclude in Section \ref{sec:concl}.

\begin{table}[h]
\centering
	\caption{Summary of the results in Sections \ref{ex2n} and \ref{sec:3x3}. Here, $\l$ denotes the non-increasing ordered eigenvalue vector of a state.}
	\label{tab}
\begin{tabular}{c|c}
	\hline
	\textbf{Sets} & \textbf{Necessary and sufficient condition for determining extreme points} \\ \hline
	\multirow{2}{*}{
	$\as_{2,2}$ ($\app_{2,2}$)} &  $\lambda_1=\lambda_{3}+2\sqrt{\lambda_{2}\lambda_{4}}$,\\ & at least two of $\l_1,\l_2,\l_3,\l_4$ are equal.   \\ \hline
	\multirow{2}{*}{
		$\as_{2,n}$ ($\app_{2,n}$)}& $\diag(\l_1,\l_{2n-2},\l_{2n-1},\l_{2n})$ is an (unnormalized) extreme point of $\as_{2,2}$,\\ & $\lambda_i\in \{\lambda_1,\lambda_{2n-2}\}$
	for any $2\le i\le 2n-3$.\\ 
	\hline
	\multirow{2}{*}{
		$\app_{3,3}$} & at least one of $l_1(\l),l_2(\l)$ defined in (\ref{3e2}) equals to zero,\\ 
	&$\l$ satisfies the criterion proposed in Theorem \ref{judge}.\\
	\hline 
	\multirow{2}{*}{
		$\app_{3,n}$} & $\diag(\l_1,\l_2,\l_3,\l_{3n-5},\cdots,\l_{3n})$ is an (unnormalized) extreme point of $\app_{3,3}$, \\
	& $\l_i\in \{\l_3,\l_{3n-5}\}$ for any $i=4,\cdots,3n-6$.\\
\hline
\end{tabular}
\end{table}

\section{Preliminaries}
\label{sec:pre}

We first introduce the notations used in this paper. We refer a quantum state $\r\in M_m(\bbC)\otimes M_n(\bbC)$ as an $m\times n$ state.
For convenience, we also work with unnormalized states, and it will be clear from the context whether we require the states to be normalized.
We denote $\diag(a_1,\cdots,a_n)$ as the order-$n$ diagonal matrix whose $j$-th diagonal entry is $a_j$. 
Given an order-$n$ matrix $M$, we denote $\diag(M)$ as the order-$n$ diagonal matrix by vanishing all the non-diagonal entries of $M$. Suppose $\cK$ is a subset of $\{1,\cdots,n\}$ with $k$ elements, we write $M_{\cK}$ for the order-$k$ principal submatrix of $M$ that corresponds to the rows and columns with index in $\cK$. We denote $||M||_p:=[\tr(M^\dagger M)^\frac{p}{2}]^\frac{1}{p}$ as the Schatten $p$-norm of $M$. We denote $\cH_{n}$ and $\cS_{n}$ as the space of order-$n$ Hermitian matrices and real symmetric matrices, respectively. Given $M\in \cH_{n}$, we refer to $\l(M):=(\l_1(M),\cdots,\l_{n}(M))$ as the eigenvalue vector of $M$, arranged in non-increasing order. We shall take $\l$ as $\l(M)$ and $\l_j$ as $\l_j(M)$ when $M$ is clear from the context. We write $M\ge \mathcal{O}$ (resp. $M>\mathcal{O}$) if $M$ is positive semidefinite (resp. definite).
Given a vector $x\in \bbR^n$, we rearrange the components of $x$ and obtain $x^{\downarrow}:=(x_1^{\downarrow},\cdots,x_n^{\downarrow})$, where $x_1^{\downarrow}\ge \cdots\ge x_n^{\downarrow}$. Similarly, $x^{\uparrow}$ denotes the vector where the components of $x$ are in non-decreasing order. Given two vectors $x,y\in \bbR^n$, we say that $y$ majorizes $x$, denote as $y\succ x$ (or $x\prec y$)  if 
\begin{eqnarray}
	\notag
&&\sum_{i=1}^k x_i^{\downarrow}\le \sum_{i=1}^k y_i^{\downarrow}, k=1,\cdots,n-1,\\
\notag
&&\sum_{i=1}^n x_i^{\downarrow}=\sum_{i=1}^n y_i^{\downarrow}.
\end{eqnarray}
Given $A,B\in \cH_n$, we say that $A\succ B$ if $\lambda(A)\succ \lambda(B)$. 
It is a well-known result that $\l(A)+\l(B)\succ\l(A+B)$.
The Schur Theorem states that for any $M\in \cH_n$, $M\succ \diag(M)$. 
There are also some facts that we shall use in the following, which can be easily proven. For more details, we refer readers to \cite[Section.10]{2011Matrix}.

\begin{lemma}
	\label{hjk}
	(i) Let $x,y\in \bbR^n$. Then
$
x+y\succ x^{\downarrow}+y^{\uparrow}.
$

	(ii) Let $x,y\in \bbR^m$, $u,v\in \bbR^n$ satisfy that $x\succ y$ and $u\succ v$. Then $(x,u)\succ (y,v)$.
	
	(iii) Let $x\in \bbR^n$ with $x_1^{\downarrow}=\cdots=x_k^{\downarrow}\ge x_{k+1}^{\downarrow}=\cdots=x_{n}^{\downarrow}$ for $1\le k\le n$. If $y\in \bbR^n$ satisfies $\sum_{i=1}^k y_i^{\downarrow}\ge\sum_{i=1}^k x_i^{\downarrow}$ and $\sum_{i=1}^n y_i^{\downarrow}=\sum_{i=1}^n x_i^{\downarrow}$, then $y\succ x$.
\end{lemma}

 A real-valued function $f$ defined on a set $\cA \subset \bbR^n$
 is said to be Schur-concave on $\cA$ if $x\succ y$ implies $f(x)\le f(y)$. If, in addition, $f(x)<f(y)$ whenever $x\succ y$ but not a permutation of $y$, then $f$ is said to be strictly Schur-concave on $\cA$. 
Denote $\cZ_n:=\{x\in \bbR^n:x_1\ge\cdots\ge x_n\ge 0\}$ and $\cZ_n^\circ:=\{x\in \bbR^n:x_1>\cdots>x_n>0\}$ as the interior of $\cZ_n$. The following lemma is a corollary of Schur-Ostrowski Theorem (see \cite[Section 3]{marshall1979inequalities}). 
\begin{lemma}
	\label{scco}
Let $f(x)$ be a real-valued function, defined and
continuous on $\cZ_n$ and continuously differentiable on $\cZ_n^\circ$.
If $\frac{\partial f(z)}{\partial z_1}<\cdots<\frac{\partial f(z)}{\partial z_n}$ holds for any $z\in \cZ_n^\circ$, then $f(x)$ is strictly Schur-concave on $\cZ_n$.
\end{lemma} 
 
A state $\r\in \as_{m,n}$ (resp. $\app_{m,n}$) is called an extreme point if $\r=t\r_1+(1-t)\r_2$ for any $t\in (0,1)$ and $\r_1,\r_2\in \as_{m,n}$ (resp. $\app_{m,n}$) implies that $\r_1=\r_2$, or equivalently, $\r_1$ and $\r_2$ are linearly dependent. 
Since $\as_{m,n}\subseteq \app_{m,n}$, it follows that any extreme point of   
$\app_{m,n}$ that belongs to $\as_{m,n}$ is also an extreme point of $\as_{m,n}$. Next, one can verify that if $\r$ is an extreme point of $\as_{m,n}$ (resp. $\app_{m,n}$), then $U\r U^\dg$ is also an extreme point of $\as_{m,n}$ (resp. $\app_{m,n}$) for any unitary matrix $U$. This is due to the fact the linearly dependence of two quantum states remains unchanged under unitary operation. Hence, whether a state in $\as_{m,n}$ is an extreme point relies on its eigenvalues
only, allowing us to consider the state in diagonal form without loss of generality. Moreover, the state $\r\in \as_{m,n}$ (resp. $\app_{m,n}$) is called an 
interior point, if there exists  $\epsilon>0$ such that $\frac{1}{1-\epsilon}(\r-\epsilon \frac{1}{mn}I_{mn})\in \as_{m,n}$ (resp. $\app_{m,n}$). Otherwise, $\r$ is called a boundary point.
By definition, any extreme point of $\as_{m,n}$ (resp. $\app_{m,n}$) is necessarily a boundary point.

We next summarize some properties about $\as_{m,n}$ and $\app_{m,n}$.
Firstly, there is a ball of AS states centered at the maximally mixed state, which is known as the {\em maximal ball}. The exact size of such ball is characterized as follows:
\begin{lemma}[\cite{gurvits2002largest}]
	\label{le:I+X=AS}
If $X\in \cH_{mn}$ satisfies $||X||_2\le 1$, then $I_{mn}+X\in \as_{m,n}$ (unnormalized).
In particular, if the $m\times n$ state $\r$ satisfies $\tr(\rho^2)\le \frac{1}{mn-1}$, then $\rho\in \as_{m,n}$. 
\end{lemma}

\begin{lemma}
	\label{le:rankrhoA=m}	
 Suppose $\r\in\app_{m,n}$. Then

(i) \begin{eqnarray}
	\label{eq:e1}
	\lambda_1\le \min\{\lambda_{mn-1}+2\sqrt{\lambda_{mn-2}\lambda_{mn}}
,\frac{3}{2+mn}\}.
\end{eqnarray}
In particular, $\l_1=\frac{3}{2+mn}$ if and and only if $\l=(\frac{3}{2+mn},\frac{1}{2+mn},\cdots,\frac{1}{2+mn})$. In this case,  $\r\in \as_{m,n}$, and it is an extreme point of both $\app_{m,n}$ and $\as_{m,n}$.

(ii)  $\r$ has deficient rank if and only if $\l(\r)=({1\over mn-1},\cdots,{1\over mn-1},0)$. In this case, $\r\in \as_{m,n}$, and it is an extreme point of both $\app_{m,n}$ and $\as_{m,n}$.
\end{lemma}
\begin{proof}
(i) The inequality $\l_1\le \lambda_{mn-1}+2\sqrt{\lambda_{mn-2}\lambda_{mn}}$ follows from \cite{hildebrand2007positive}. The inequality $\l_1\le \frac{3}{2+mn}$ and subsequent claim are obtained from \cite[Proposition 8.2]{2015Positive}. To demonstrate that the state $\r:=\diag(\frac{3}{2+mn},\frac{1}{2+mn},\cdots,\frac{1}{2+mn})$ is an extreme point of $\app_{m,n}$, we assume the contrary, i.e., $\rho=p\a+(1-p)\b$, where $p\in (0,1)$ and $\a,\b\in \app_{m,n}$ are linearly independent. We have $\frac{3}{2+mn}\le p\l_1(\a)+(1-p)\l_1(\b)$. It follows from (\ref{eq:e1}) that $\l_1(\a)=\l_1(\b)=\frac{3}{2+mn}$ and thus $\l(\a)=\l(\b)=(\frac{3}{2+mn},\frac{1}{2+mn},\cdots,\frac{1}{2+mn})$. Consequently, by Schur Theorem, the first diagonal entries of both $\a,\b$ are $\frac{3}{2+mn}$, with the remaining entries in the first row and column 
being zero. Since the last $mn-1$ eigenvalues of $\a,\b$ are identical, 
the last order-$(mn-1)$ principal submatrices of $\a,\b$  are proportional to the identity
matrix. This leads to that $\a=\b=\r$, which contradicts the initial assumption. Hence $\r$ is an extreme point of $\app_{m,n}$, and consequently, an extreme point of $\as_{m,n}$.

(ii) The first claim follows from \cite[Proposition 1]{arunachalam2014absolute}. The claim that $\rho\in \as_{m,n}$ follows from Lemma \ref{le:I+X=AS}. Suppose $\rho:=\diag(\frac{1}{mn-1},\cdots,\frac{1}{mn-1},0)$ is a non-extreme point of $\app_{m,n}$, that is, $\rho=p\g+(1-p)\eta$, where $p\in (0,1)$ and $\a,\b\in \app_{m,n}$ are linearly independent. We have $0=\lambda_{mn}(\rho)\ge p\lambda_{mn}(\g)+(1-p)\lambda_{mn}(\eta)$. This implies that $\lambda_{mn}(\g)=\lambda_{mn}(\eta)=0$, leading to $\l(\g)=\l(\eta)=(\frac{1}{mn-1},\cdots,\frac{1}{mn-1},0)$. By employing a similar approach as in (i), one can verify that $\g=\eta=\r$. This is a contradiction. So $\r$ is an extreme point of $\app_{m,n}$, and of $\as_{m,n}$.
  \end{proof}

\begin{theorem}
	\label{maj}
	Suppose the two states $\sigma,\r$ satisfy $\sigma\in \as_{m,n}$ (resp. $\app_{m,n}$) and $\sigma\succ \r$. Then $\r\in \as_{m,n}$ (resp. $\r\in\app_{m,n}$). Further, if $\l(\r)\neq \l(\sigma)$, then $\r$ is a non-extreme point of $\as_{m,n}$ (resp. $\app_{m,n}$). 
\end{theorem}
\begin{proof}
	The first fact follows from \cite[Lemma 2.2]{2015Positive}.
	Since $\sigma\succ \r$, by Uhlmann’s theorem \cite{alberti1982stochasticity}, there exist unitary matrices $U_j$ and a probability distribution $\{p_j\}$ such that $\r=\sum_{j}p_jU_j\sigma U_j^\dagger$. Let $\rho$ and $\sigma$ have distinct eigenvalues. If $\r$ is an extreme point of $\as_{m,n}$, then it is linearly dependent with $U_j\sigma U_j^\dg$, and thus $\l(\r)=\l(\sigma)$. This is a contradiction. 
\end{proof}

\begin{theorem}
\label{le:mxn=AS}
If $\r\in \as_{m,n}$ (resp. $\app_{m,n}$), then the unnormalized state $\r_{\cK}\in \as_{p,q}$ (resp. $\app_{p,q}$) for any $\cK=\{i_1,i_2,...,i_{pq}\}\subseteq \{1,2,...,mn\}$, where $1<p\le m$ and $1<q\le n$.
\end{theorem}
\begin{proof}
We prove the claim for $\as_{m,n}$, the claim for $\app_{m,n}$ can be proved similarly.  Let 
$\cK'=\cup_{k=0}^{p-1}\{kn+1,\cdots,kn+q\}$. Since $\r\in \as_{m,n}$, by permuting the rows and columns of $\r$, we obtain another state $\r'\in\as_{m,n}$, which satisfies that $\r'_{\cK'}=\r_{\cK}$. We next prove that $\r'_{\cK'}\in \as_{p,q}$.

Given any order-$pq$ unitary matrix $V$, let the corresponding order-$mn$ unitary matrix $U$ such that $U_{\cK'}=V$ and $U_{(\cK')^c}=I_{mn-pq}$. 
We have $U\r' U^\dg$ is separable, i.e.,  it can be written as $\sum_j p_jA_j\otimes B_j$, where $p_j$ is a probability distribution, $A_j\ge \mathcal{O}$ has order $m$ and 
$B_j\ge \mathcal{O}$ has order $n$. Let $\cK_1'=\{1,\cdots,p\}$ and $\cK_2'=\{1,\cdots,q\}$. Through direct matrix computation, one can verify that $V\r'_{\cK'}V^\dg=\sum_jp_j ({A_j})_{\cK_1'}\otimes ({B_j})_{\cK_2'}$, where 
$({A_j})_{\cK_1'}\ge \mathcal{O}$ has order $p$ and $({B_j})_{\cK_2'}\ge \mathcal{O}$ has order $q$. This implies that $V\r'_{\cK'}V^\dg$ is separable. Thus $\r'_{\cK'}\in\as_{p,q}$. This completes the proof.
\end{proof}

{\bf Remark.} The converse of the above claim may not hold. An example is $\r=\frac{1}{84}\diag(15,14,9,9,9,9,9,9,1)$. One can verify that $\r_{\cK}\in\app_{2,2}$ for any $\cK=\{i_1,i_2,i_3,i_4\}\subseteq\{1,\cdots,9\}$, but $\r\notin \app_{3,3}$ according to the criterion in \cite{hildebrand2007positive}.

In the final part of this section, we present a property of non-extreme points in $\as_{m,n}$ and $\app_{m,n}$. The proof of the following theorem will be given in Appendix \ref{pt1012}.

\begin{theorem}
	\label{thm:main}	
	Suppose the diagonal state $\r\in \as_{m,n}$ (reps. $\app_{m,n}$) is a non-extreme point. Then 
	
	(i) there exist two linearly independent diagonal states $\a,\b\in \as_{m,n}$ (reps. $\app_{m,n}$) such that $\r=t\a+(1-t)\b$ for $t\in (0,1)$. 
	
	(ii) for any $\epsilon>0$, there exist two linearly independent diagonal states $\a',\b'\in \as_{m,n}$ (reps. $\app_{m,n}$) such that $\r=\frac{1}{2}(\a'+\b')$, where $||\r-\a'||_2=||\r-\b'||_2<\epsilon$.
\end{theorem}
\section{Extreme Points of $\as_{2,n}$ and $\app_{2,n}$}
\label{ex2n}
In this section, we provide a characterization of extreme points for both $\as_{2,n}$ and $\app_{2,n}$, as these two sets are identical. Recall that $\r \in \as_{2,n}$ if and only if 
\begin{eqnarray}
	\label{eq:2xnAS}	
	\lambda_1\le \lambda_{2n-1}+2\sqrt{\lambda_{2n-2}\lambda_{2n}},
\end{eqnarray}
or equivalently,
\begin{eqnarray}
	\label{drt}
	\bma 2\l_{2n}&&\l_{2n-1}-\l_1\\\l_{2n-1}-\l_1&&2\l_{2n-2}\ema\ge \mathcal{O}.
	\end{eqnarray}

A characterization of $\as_{2,n}$ asserts that
 $\r\in \as_{2,n}$ is a boundary point if the inequality (\ref{eq:2xnAS}) is saturated, or equivalently, the matrix in (\ref{drt}) has rank one,
otherwise $\r$ is an interior point \cite{halder2021characterizing}. In the same article, it has also been shown that there exist boundary points of $\as_{2,n}$ that are non-extreme. Here, we first present another property about boundary points of $\as_{2,2}$. 
\begin{theorem}
	\label{plpl}
	Let $\r$ be a boundary point of $\as_{2,2}$. If the state $\sigma\succ \r$ with $\l(\sigma)\neq \l(\r)$, then $\sigma \notin \as_{2,2}$.
\end{theorem}
\begin{proof}
	Define the real-valued function $f$ on $\cZ_4$ by $f(x):=x_{3}+2\sqrt{x_{2} x_{4}}-x_{1}$. Since $\r$ is a boundary point, we have $f(\l(\r))=0$. Notice that $f$ is continuous and continuously differentiable on $\cZ_4^\circ$. Further, for any $z\in \cZ_4^\circ$, we have $\frac{\partial f(z)}{\partial z_1}=-1, \frac{\partial f(z)}{\partial z_2}=\sqrt{\frac{z_{4}}{z_2}},\frac{\partial f(z)}{\partial z_3}=1$, and 
		$\frac{\partial f(z)}{\partial z_4}=\sqrt{\frac{z_{2}}{z_4}},
	$
	which implies $\frac{\partial f(z)}{\partial z_1}<\cdots<\frac{\partial f(z)}{\partial z_4}$. By using Lemma \ref{scco} on $\cZ_4$, we obtain that $f$ is strictly Schur-concave. 
Since $\l(\sigma)\succ \l(\r)$, where the two vectors are distinct, we have $f(\l(\sigma))<f(\l(\r))=0$. This implies that $\sigma\notin\as_{2,2}$ as it violates (\ref{eq:2xnAS}).
\end{proof}

{\bf Remark.} The above claim may not hold for the boundary point $\r\in\as_{2,n}$. For instance, consider $\sigma=\frac{1}{10}\diag(3,3,1,1,1,1)\in \as_{2,3}$ and $\r=\frac{1}{10}\diag(3,2,2,1,1,1)\in \as_{2,3}$.

We now characterize the extreme points of $\as_{2,2}$ in Theorem \ref{ef22}, and then generalize the results to $\as_{2,n}$ in Theorem \ref{mmain}. The detailed proof of the following two theorems are contained in Appendix \ref{mtp}.

\begin{theorem}
	\label{ef22}
	The state $\rho\in \as_{2,2}$ ($\app_{2,2}$) is an extreme point if and only if the following two conditions hold:
	
	(I) $\r$ is a boundary point, i.e., $\lambda_1=\lambda_{3}+2\sqrt{\lambda_{2}\lambda_{4}}$, 
	
	(II) at least two of $\l_1,\l_2,\l_3,\l_4$ are equal.
\end{theorem}

\begin{theorem}
	\label{mmain}
The state $\rho\in \as_{2,n}$ ($\app_{2,n}$) is an extreme point if and only if the following three conditions hold:

(I) $\rho$ is a boundary point, i.e, $\lambda_1=\lambda_{2n-1}+2\sqrt{\lambda_{2n-2}\lambda_{2n}}$,

(II) $\lambda_i\in \{\lambda_1,\lambda_{2n-2}\}$
for any $2\le i\le 2n-3$,

(III) at least two of $\lambda_1,\lambda_{2n-2},\lambda_{2n-1},\lambda_{2n}$ are equal.  
\end{theorem}

{\bf Remark.} Conditions (I)-(III) imply that every extreme point of $\as_{2,n}$ has at most three distinct eigenvalues. The above theorem can also be expressed as follows: the state $\r\in\as_{2,n}$ is an extreme point if and only if $\lambda_i\in \{\lambda_1,\lambda_{2n-2}\}$
for any $2\le i\le 2n-3$, and the (unnormalized) state $\diag(\l_1,\l_{2n-2},\l_{2n-1},\l_{2n})$ is an extreme point of $\as_{2,2}$.

Recalling from the definition of maximal ball in Lemma \ref{le:I+X=AS},  we say that the state $\r$ resides on (or outside) the maximal ball if $\tr(\r^2)\le\frac{1}{mn-1}$(or $\tr(\r^2)>\frac{1}{mn-1}$).
The following result characterizes the relationship between the maximal ball and extreme points of $\as_{2,n}$. The proof is also given in Appendix \ref{mtp}.

\begin{corollary}
	\label{mbn}
	(i) Let $\r$ be an extreme point of $\as_{2,2}$. Then $\r$ resides on the maximal ball if and only if $\l(\r)=(\frac{1}{3},\frac{1}{3},\frac{1}{3},0)$ or $(\frac{1}{2},\frac{1}{6},\frac{1}{6},\frac{1}{6})$.
	
	(ii) Let $\r$ be an extreme point of $\as_{2,n}$ ($n>2$). Then $\r$ resides on the maximal ball if and only if  $\l(\r)=(\frac{1}{2n-1},\cdots,\frac{1}{2n-1},0)$.
\end{corollary}

\section{Extreme points of $\app_{3,n}$}
\label{sec:3x3}
In this section, we investigate the extreme points of $\app_{3,n}$. Our focus will primarily be on the full-rank extreme points, as the deficient-rank case has already been addressed in Lemma \ref{le:rankrhoA=m} (ii).  Similar to the structure of Section \ref{ex2n}, we begin with the two-qutrit system, and then generalize the results to the qutrit-qudit system.

Let
$
\cZ_9^{+}:=\{x\in \bbR^9:x_1\ge\cdots\ge x_9>0\}.
$
Define the linear maps $L_1, L_2:\cZ_9^{+}\rightarrow \cS_3$ and functions $l_1, l_2:\cZ_9^{+}\rightarrow \bbR$ as 
\begin{eqnarray}
	\notag
	\label{3e1}
	&&L_1(x):=\bma 2x_9&&x_8-x_1&&x_6-x_2\\x_8-x_1&&2x_7&&x_5-x_3\\x_6-x_2&&x_5-x_3&&2x_4\ema, \quad l_1(x):=\det L_1(x),\\
\label{3e2}
	&&L_2(x):=\bma 2x_9&&x_8-x_1&&x_7-x_2\\x_8-x_1&&2x_6&&x_5-x_3\\x_7-x_2&&x_5-x_3&&2x_4\ema, \quad l_2(x):=\det L_2(x).
	\end{eqnarray}
	It is known that 
the full-rank state $\rho$ belongs to $\app_{3,3}$ if and only if $L_1(\l)\ge \mathcal{O}$ and
$L_2(\l)\ge \mathcal{O}$ \cite{hildebrand2007positive}. We first propose some properties of boundary points in $\app_{3,3}$. 

\begin{lemma}
	\label{le:L1+L21}
If $\r$ is a boundary point of $\app_{3,3}$, then at least one of $l_1(\l),l_2(\l)$ equals to zero.
\end{lemma}
 \begin{proof}
	We know that $l_1(\l),l_2(\l)\ge 0$. Assume that 
	$l_1(\l), l_2(\l)>0$, and consequently, $L_1(\l),L_2(\l)>\mathcal{O}$. There exists a small enough $\epsilon>0$ such that $L_1(\l)-\frac{2\epsilon}{9} I_3>\mathcal{O}$ and $L_2(\l)-\frac{2\epsilon}{9} I_3>\mathcal{O}$. Let $\sigma=\frac{1}{1-\epsilon}(\r-\epsilon \frac{1}{9}I_9)$. It follows that $L_1(\sigma),L_2(\sigma)>\mathcal{O}$, thus $\sigma\in \app_{3,3}$. This implies that $\r$ is an interior point, which is a contradiction. 
\end{proof}

The proof of the following two theorems will be given in Appendix \ref{profL1+L2}.

\begin{theorem}
	\label{3bb}
Let $\r$ be a boundary point of $\app_{3,3}$. If the state 
		$\sigma \succ \rho$ with $\l(\sigma)\neq \l(\r)$,
		then $\sigma\notin \app_{3,3}$.
\end{theorem}

\begin{theorem}
	\label{nonbou}
	Let $\r=\diag(\l_1,\cdots,\l_9)$ be a boundary point of $\app_{3,3}$. Then $\r$ is a non-extreme point if and only if $\r=\frac{1}{2}(\a+\b)$, where $\a,\b\in \app_{3,3}
$ are two linearly independent diagonal states whose diagonal entries are both in non-increasing order.
\end{theorem}

We are prepared to provide a necessary and sufficient condition for identifying whether a boundary point of $\app_{3,3}$ is an extreme point.
Subsequently, we investigate more properties of extreme points in $\app_{3,3}$. The proof of the following theorem and corollary will also be detailed in Appendix \ref{profL1+L2}.

\begin{theorem}
	\label{judge}
	Let $\r$ be a full-rank boundary point of $\app_{3,3}$. 
	Let $U:=\bma U_1&U_2&U_3\ema$ and $V:=\bma V_1&V_2&V_3\ema$ (written as column vectors) be the order-three real orthogonal matrices 
	such that
\begin{eqnarray}
	\notag
	\label{D1}
	&&U^T L_1(\l) U= D_1,\\
	\label{D2}
	&&V^T L_2(\l) V= D_2,
\end{eqnarray}
	where $D_1,D_2\ge \mathcal{O}$ are diagonal matrices with diagonal entries in non-increasing order.
	
	(i) Suppose $l_1(\l)=0, l_2(\l)>0$. 
Then $\r$ is an extreme point if and only if the following linear equations in terms of $t_1,\cdots,t_9$ have only the trivial solution (i.e., $t_i$ are all equal to zero):
	\begin{eqnarray}
		\label{fn1}
  &&\sum_{i=1}^9 t_i=0,\\
   \label{fn3}
  &&t_k=t_{k+1} \text{ whenever }  \l_k=\l_{k+1},\\
  \label{fn2}
  &&\bma 2t_9&&t_8-t_1&&t_6-t_2\\t_8-t_1&&2t_7&&t_5-t_3\\t_6-t_2&&t_5-t_3&&2t_4\ema \cdot U_3=\bma 0\\0\\0\ema.
\end{eqnarray}

(ii) Suppose $l_1(\l)>0, l_2(\l)=0$. 
Then $\r$ is an extreme point if and only if Eqs.(\ref{fn1}), (\ref{fn3}), and the following linear equation in terms of $t_1,\cdots,t_9$ have only the trivial solution:
\begin{eqnarray}
\label{kn2}
	&&\bma 2t_9&&t_8-t_1&&t_7-t_2\\t_8-t_1&&2t_6&&t_5-t_3\\t_7-t_2&&t_5-t_3&&2t_4\ema \cdot V_3=\bma 0\\0\\0\ema.
\end{eqnarray}

(iii) Suppose $l_1(\l)=0, l_2(\l)=0$.
Then $\r$ is an extreme point if and only if Eqs.(\ref{fn1})-(\ref{kn2}) in terms of $t_1,\cdots,t_9$ have only the trivial solution.
\end{theorem}

{\bf Remark.} Eqs.(\ref{D1}) and (\ref{D2}) remain valid when substituting the column vectors $U_3,V_3$ with $-U_3,-V_3$, respectively. But Eqs.(\ref{fn2}) and (\ref{kn2}) guarantee that selecting either one is adequate.
\begin{corollary}
	\label{ccll}
(i) Suppose the state $\r\in \app_{3,3}$ has exactly two distinct eigenvalues. Then $\r$ is an extreme point if and only if it is unitarily equivalent to one of the following eight states:
	\begin{eqnarray}
	\label{twodi}
	\notag
	&&\zeta_1=\frac{1}{11}\diag(3,1,1,1,1,1,1,1,1),\\
	\notag
	&&\zeta_2=\frac{1}{9+2\sqrt{2}}\diag(\sqrt{2}+1,\sqrt{2}+1,1,1,1,1,1,1,1),\\
	\notag
	&&\zeta_3=\frac{1}{12}\diag(2,2,2,1,1,1,1,1,1),\\
	\notag
	&&\zeta_4=\frac{1}{10+\sqrt{17}}\diag(\frac{5+\sqrt{17}}{4},\frac{5+\sqrt{17}}{4},\frac{5+\sqrt{17}}{4},\frac{5+\sqrt{17}}{4},1,1,1,1,1),\\
	\notag
	&&\zeta_5=\frac{1}{5x+4}\diag(x,x,x,x,x,1,1,1,1),\\
	\notag
	&&\zeta_6=\frac{1}{21}\diag(3,3,3,3,3,3,1,1,1),\\
	\notag
	&&\zeta_7=\frac{1}{23+14\sqrt{2}}\diag(3+2\sqrt{2},\cdots,3+2\sqrt{2},1,1),\\
	&&\zeta_8=\frac{1}{8}\diag(1,1,1,1,1,1,1,1,0),
\end{eqnarray}
where $x\approx 2.70928$ is a root of the equation $x^3-x^2-5x+1=0$.

(ii) The boundary point $\r\in \app_{3,3}$ with  $\l(\r)=(a,b,c,c,c,c,c,c,c)$ ($a>b>c$) is an extreme point.

(iii) Every extreme point of $\app_{3,3}$ has at most seven distinct eigenvalues. 
\end{corollary}

From Lemma \ref{le:I+X=AS}, we know that $\z_1,\z_8\in \as_{3,3}$. We also have $\z_3\in \as_{3,3}$, as it resides on the maximal ball. But whether the other five states belong to $\as_{3,3}$ remains uncertain. We can further derive another property of the states in $\app_{3,3}$, specifically regarding the upper bound for the sum of $k$ largest eigenvalues. Define $\Lambda_k(\app_{3,3}):=\sup_{\r\in\app_{3,3}} \sum_{i=1}^{k}\l_i(\r)$.
\begin{corollary}
For any $1\le k\le 8$, $\Lambda_k(\app_{3,3})$ is attained at the state that is unitarily equivalent to $\z_k$ in (\ref{twodi}).
\end{corollary}
\begin{proof}
For each $1\le k\le 8$, assume that there exists a state $\sigma\in \app_{3,3}$ such that 
$\sum_{i=1}^k \l_i(\sigma)>\sum_{i=1}^k \l_i(\z_k)$. Using Lemma \ref{hjk} (iii),  we obtain that $\sigma\succ \z_k$,  where $\l(\sigma)\neq \l(\zeta_k)$. Noting that $\z_k$ is a boundary point of $\app_{3,3}$, from Theorem \ref{3bb}, we have $\sigma\notin\app_{3,3}$. This contradicts with the assumption. Thus the claim is proved. 
\end{proof}

Finally, we consider the extreme points of $\app_{3,n}$. It is known that the state $\r$ belongs to $\app_{3,n}$ if and only if 
\begin{eqnarray}
	\notag
	\label{wym}
	&&L_1((\l_1,\l_2,\l_3,\l_{3n-5},\cdots,\l_{3n}))\ge \mathcal{O},\\
	\label{wyym}
	&&L_2((\l_1,\l_2,\l_3,\l_{3n-5},\cdots,\l_{3n}))\ge \mathcal{O}.
\end{eqnarray}
The results in $\app_{3,3}$ can be extended to provide a necessary and sufficient condition for an extreme point in $\app_{3,n}$, as outlined in the following theorem. The proof is also contained in Appendix \ref{profL1+L2}.
\begin{theorem}
	\label{3ndpd}
	The state $\r\in\app_{3,n}$ ($n>3$) is an extreme point if and only if it satisfies the following two conditions:
	
	(I) the unnormalized state $\diag(\l_1,\l_2,\l_3,\l_{3n-5},\cdots,\l_{3n})$ is an extreme point 
	of $\app_{3,3}$, 
	
	(II) $\l_i\in \{\l_3,\l_{3n-5}\}$ for any $i=4,\cdots,3n-6$.
\end{theorem}

{\bf Remark.} Combining with Corollary \ref{ccll} (iii), we obtain that each extreme point in $\app_{3,n}$ also has at most seven distinct eigenvalues.

At the end of this section, we point out that the above criterion of extreme points in qutrit-qudit system may be extended to higher-dimensional systems. 
However, one challenge is that the number of linear maps needed to determine whether a point belongs to $\app_{m,n}$ increases exponentially with the dimensions of the system (see \cite{johnston2018inverse}), which raises the complexity of the computations.

\section{Robustness of nonabsolute separability}
\label{ase}
In this section, we introduce the concept of robustness of nonabsolute separability. We first review the robustness of entanglement, which 
was initially considered by Vidal and Tarrach to quantify entanglement \cite{Vidal1999Robustness}. Given a quantum state $\r$, the robustness of entanglement for $\r$ is the minimum value $t$ such that $\frac{1}{1+t}(\r+t\sigma)$ is separable, where $\sigma$ denotes an arbitrary separable state. The authors also constrained the state $\sigma$  to be white noise (the identity operator), denoting this as random robustness, and proved that for any $m\times n$ state $\r$, the unnormalized state $\r+\frac{1}{2}I_{mn}$ is always separable. It is worth noting that 
$\r+\frac{1}{2}I_{mn}$ is also AS, since it remains separable under any global unitary operations. Later, the concept was generalized where $\sigma$ can denote any quantum state \cite{Steiner2003Generalized}, termed as generalized robustness of entanglement, denoted as $\cR(\r)$. For pure states $\r$, the robustness of entanglement and $\cR(\r)$ are exactly characterized, which are both equal to $(\sum_{i=1} a_i)^2-1$,
where $a_i$ is the Schmidt coefficients of $\ket{\phi}$ with $\r=\proj{\phi}$ \cite{Steiner2003Generalized,Harrow2003RobustnessOQ}.
But the question of whether these two values are equal for mixed states remains unresolved. Generalized robustness of entanglement offers a geometrical approach to quantify entanglement by measuring the distance to the set of separable states \cite{2006Connecting}.

On the other hand, a resource theory of NAS was developed, where mixtures of global unitary operations are considered free operations, while AS states are free states \cite{2023Resource}. In this resource theory,
a ''good'' NAS measure should satisfy criteria including positivity, invariance with local unitary operations, monotonicity under
free operations, and convexity. In \cite{2023Resource}, several good measures have been introduced, utilizing distance measures like relative entropy, Bures distance, and Hilbert-Schmidt distance (for detailed information, see \cite{bengtsson2017geometry}) quantifying the distance of a non-AS state from the set of AS states. In the following, we introduce a new measure of NAS, analogy to the generalized robustness of entanglement, which offers an additional geometric insight into non-AS states.
\begin{definition}
	Given an $m\times n$ state $\r$, we define the robustness of nonabsolute separability of $\r$ as the minimal value $t\ge 0$ such that $\frac{1}{1+t}(\r+t\sigma)\in \as_{m,n}$, where $\sigma$ can be chosen from any $m\times n$ state. We denote this value as $\ar(\r)$, with 
	$\sigma$ termed as an optimal state.
\end{definition}

By definition, we immediately have $\cR(\r)\le \ar(\r)\le \frac{mn}{2}$ for any $m\times n$ state $\r$, where the first inequality follows from that an AS state is necessarily separable, and the second inequality is the known result of random robustness. We next investigate more properties about $\ar(\r)$.

\begin{theorem}
	\label{pofar}

(i)  The state $\r\in \as_{m,n}$ if and only if $\ar(\r)=0$.

	(ii) $\ar(\r)=\ar(U\r U^\dg)$ for any state $\r$ and global unitary matrix $U$.
	
	(iii) $\ar(\r)$ is convex.
	
	(iv) $\ar(\r)\ge \ar(\sigma)$ for any two states $\r,\sigma$ such that $\r\succ \sigma$.

	(v) Let $\r$ be an $m\times n$ state with $\ar(\r)=t>0$ and $\sigma$ be an optimal state. Then $\frac{1}{1+t}(\r+t\sigma)$ is necessarily a boundary point of $\as_{m,n}$.
	\end{theorem}
\begin{proof}
    (i) This is obvious.

	(ii) 	Let $\ar(\r)=t$ and $\sigma$ be the optimal state such that $\frac{1}{1+t}(\r+t\sigma)\in
	\as_{m,n}$. For any order-$mn$ unitary matrix $U$, we have $\frac{1}{1+t}(U\r U^\dg+tU\sigma U^\dg)\in \as_{m,n}$. So $t\ge \ar(U\r U^\dg)$. Conversely, we have $\ar(U\r U^\dg)\ge t$. Thus the claim holds.

(iii) For any two states $\r_1, \r_2$ and $p\in (0,1)$, let $\ar(\r_1)=t_1,\ar(\r_2)=t_2$ and $\sigma_1,\sigma_2$ be the optimal states respectively. Thus $\r_1+t_1\sigma_1\in \as_{m,n}$ and $\r_2+t_2\sigma_2\in \as_{m,n}$. Consequently, $p\r_1+(1-p)\r_2+(pt_1+(1-p)t_2)\sigma\in \as_{m,n}$,
where $\sigma=\frac{1}{pt_1+(1-p)t_2}(pt_1\sigma_1+(1-p)t_2\sigma_2)$. This implies that $\ar(p\r_1+(1-p)\r_2)\le pt_1+(1-p)t_2$.
	
	(iv) This can be verified from (ii) and (iii).

	(v) Assume that $\frac{1}{1+t}(\r+t\sigma)$ is an interior point of $\as_{m,n}$. There exists $\epsilon_1>0$ such that $\r+t\sigma-\epsilon_1 \frac{1}{mn}I_{mn}\in \as_{m,n}$.  Further, there exists a small enough $\epsilon_2>0$ such that $||\frac{\epsilon_2}{\epsilon_1}\sigma||_2<\frac{1}{mn}$. Using Lemma \ref{le:I+X=AS}, we have  $\frac{1}{mn}I_{mn}-\frac{\epsilon_2}{\epsilon_1}\sigma\in \as_{m,n}$. Consequently, 
$\r+t\sigma-\epsilon_2\sigma
		=(\r+t\sigma-\epsilon_1 \frac{1}{mn}I_{mn})+\epsilon_1(\frac{1}{mn}I_{mn}-\frac{\epsilon_2}{\epsilon_1}\sigma)\in \as_{m,n}$.
This implies that $\ar(\r)\le t-\epsilon_2<t$, which is a contradiction. So the assumption does not hold and the claim is proved. 
	\end{proof}

{\bf Remark.} Parts (i)-(iv) of the above theorem imply that $\ar(\r)$ is a good NAS measure as proposed in \cite{2023Resource}. Moreover, (iv) coincides with Theorem 2 of \cite{2023Resource} that any pure state possesses the maximal amount of resources. 

We next give a compact form of $\ar(\r)$ for certain states. From (ii) of the above theorem, it suffices to consider all states in diagonal form. The proof of following theorem will be given in Appendix \ref{james}.

\begin{theorem}
	\label{battle}
(i)	For the  $m\times n$ state $\r=\diag(1,0,\cdots,0)$, $\ar(\r)=\frac{mn-1}{3}$, where the unique optimal state is $\diag(0,\frac{1}{mn-1},\cdots,\frac{1}{mn-1})$.

(ii) For the $2\times n$ state $\r=\diag(\frac{1}{k},\cdots,\frac{1}{k},0,\cdots,0)$ with $1\le k\le2n-3$, $\ar(\r)=\frac{2n-k}{3k}$, where the unique optimal state is $\diag(0,\cdots,0,\frac{1}{2n-k},\cdots,\frac{1}{2n-k})$.

(iii) For the $2\times n$ state $\r=\diag(\frac{1}{2n-2},\cdots,\frac{1}{2n-2},0,0)$, $\ar(\r)=\frac{3-2\sqrt{2}}{n-1}$, where the unique optimal state is 
$\diag(0,\cdots,0,\frac{1}{2},\frac{1}{2})$.

(iv) Let $\r=\diag(a,1-a,0,0)$ with $\frac{1}{2}\le a\le 1$. If $\frac{1}{2}\le a\le \frac{3}{4}$, then $\ar(\r)=4-4\sqrt{1-a}-2a$, where the optimal state can be chosen as $\diag(0,0,\frac{1}{2},\frac{1}{2})$. If $\frac{3}{4}< a\le 1$, then $\ar(\r)=2a-1$, where the optimal state can be chosen as $\diag(0,\frac{4a-3}{6a-3},\frac{a}{6a-3},\frac{a}{6a-3})$.
\end{theorem}

{\bf Remark.} Part (i) of above theorem implies that for any pure state $\r$, 
the optimal state $\sigma$ is orthogonal to $\r$, i.e., $\tr(\r \sigma)=0$. 
Moreover, both $\sigma$ and the resulting state $\g=\frac{1}{1+t}(\r+t\sigma)$ are extreme points of $\as_{m,n}$, according to Lemma \ref{le:rankrhoA=m} (see Fig. \ref{fig:robustness} for a schematic diagram in two-qubit system). However, for mixed state, (ii) shows that the optimal state is not necessarily AS, and (iv) shows that it is also not necessarily orthogonal to this state.
Additionally, (iv) can be used to provide an upper bound of the generalized robustness of entanglement for any rank-two two-qubit state, since the resulting state is necessarily separable.

\begin{figure}[ht]
\includegraphics[width=0.50\textwidth]{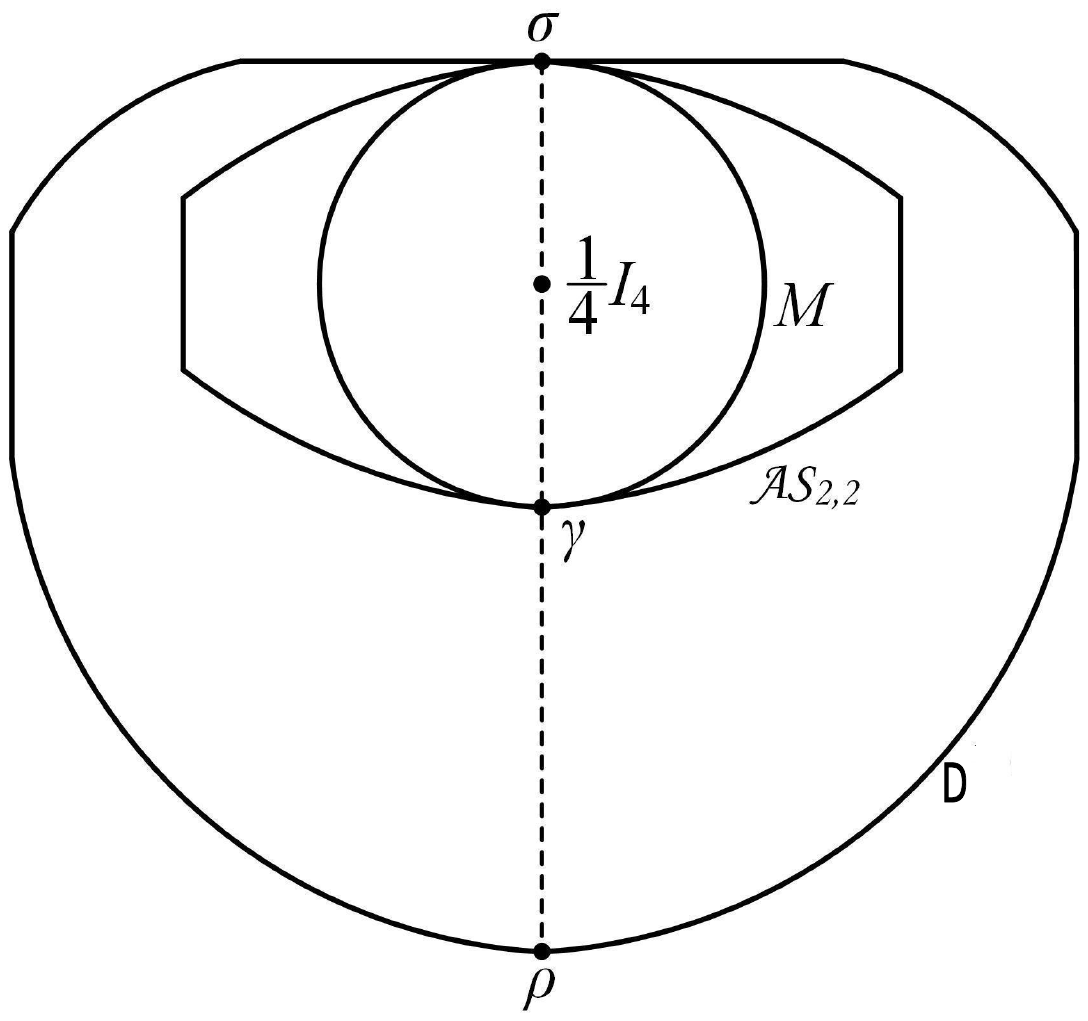}
\captionsetup{justification=justified,singlelinecheck=false}
\caption{A schematic diagram of robustness of nonabsolute separability for the pure state $\r$ in two-qubit system. Here, $M$ represents the maximal ball, and $\cD$ denotes the entire set of states. The line segments at the boundary of $\as_{2,2}$ represent those non-extreme points. The optimal state $\sigma$ is the deficient-rank extreme point of $\as_{2,2}$, residing on the boundary but not being an extreme point of $\cD$. The resulting state $\gamma$ with eigenvalues $(\frac{1}{2},\frac{1}{6},\frac{1}{6},\frac{1}{6})$ lies at the middle of $\r$ and $\sigma$. Moreover, $\sigma$ and $\gamma$ represent all the extreme points of $\as_{2,2}$ that reside on the maximal ball, as stated in Corollary \ref{mbn} (i).}
\label{fig:robustness}
\end{figure}
\vspace{-1em}

\section{Conclusion}
\label{sec:concl}
In this work, we have investigated the extreme points of sets of absolutely separable and PPT states. We have proposed some properties about the boundary of these sets. We have provided an exact form of the extreme points in $\as_{2,n}$ and $\app_{2,n}$, and established a necessary and sufficient condition for identifying extreme points in $\app_{3,n}$ for $n\ge 3$. In particular, we have found that every extreme point of $\as_{2,n}$ has no more than three distinct eigenvalues, and every extreme point of $\app_{3,n}$ has no more than seven distinct eigenvalues. We have introduced the robustness of nonabsolute separability and explored its basic properties. Specifically, we have shown that it is a good geometrical measure of a quantum state, within the resource theory of nonabsolute separability. We have provided the exact formula of pure states and some special mixed states including rank-two two-qubit states, under this measure.  

There are several open problems remaining in our research. 
The first is to present a compact form of all the extreme points of $\app_{3,n}$, similar to the results in qubit-qudit system. 
The next issue is to determine whether every extreme point of $\app_{3,n}$ belongs to $\as_{3,n}$. In particular, we conjecture this is true for all eight states in (\ref{twodi}). The challenge is that we currently lack knowledge of the criterion for a state to be classified as $\as_{3,3}$.
The third task is to compute the robustness of nonabsolute separability for general mixed states. This presents a greater challenge because of the increasing number of variables involved. Besides, one of the related problem is whether the optimal state is unique for the mixed states outside the set of AS states.

\section*{Acknowledgments}
\label{sec:ack}	
Authors were supported by the  NNSF of China (Grant No. 12471427).
	
\newpage
\appendix
\section{Proof of Theorem \ref{thm:main}}
\label{pt1012}
We prove the claim for $\as_{m,n}$, the claim for $\app_{m,n}$ can be proved similarly. 

(i) Since $\r$ is a non-extreme point of $\as_{m,n}$, there exist two linearly independent states $\sigma_1,\sigma_2\in \as_{m,n}$ such that $\r=p\sigma_1+(1-p)\sigma_2$ for $p\in (0,1)$. So $\r=p\diag(\sigma_1)+(1-p)\diag(\sigma_2)$. By using Schur Theorem and Theorem \ref{maj}, we have $\diag(\sigma_1),\diag(\sigma_2)\in \as_{m,n}$.
Hence the claim holds if $\diag(\sigma_1)$ and $\diag(\sigma_2)$ are linearly independent.

Suppose $\diag(\sigma_1)$ and $\diag(\sigma_2)$ are linearly dependent. 
This implies that $\sigma_1,\sigma_2$ are both non-diagonal with $\diag(\sigma_1)=\diag(\sigma_2)=\r$. Let $\diag(\l(\sigma_1)):=\diag(\l_1(\sigma_1),\cdots,\l_{mn}(\sigma_1))$.
We have $\diag(\l(\sigma))\in \as_{m,n}$ and by Schur Theorem, $\diag(\l(\sigma_1))\succ \r$, where the two states do not have the same eigenvalues. According to Uhlmann's Theorem, there exist permutation matrices $P_j$ and probabilities $p_j$ such that
$
\r=\sum_j p_j P_j\diag(\l(\sigma_1))P_j^\dg,
$
where at least two $p_j\in (0,1)$. Without loss of generality, let $p_1\in (0,1)$. Let $\a=P_1\diag(\l(\sigma_1))P_1^\dg$ and $\b=\frac{1}{1-p_1}(\sum_{j=2}p_j P_j\diag(\l(\sigma_1))P_j^\dg)$. 
We have $\r=t\a+(1-t)\b$, where $\a,\b\in \as_{m,n}$ are both diagonal. Further, $\a,\b$ are linearly independent, otherwise $\diag(\l(\sigma_1))$ and $\r$ would have the same eigenvalues. This completes the proof.

(ii) From (i), we can write $\r=t\a+(1-t)\b$ for $t\in (0,1)$, where $\a,\b\in \as_{m,n}$ are linearly independent diagonal states. Without loss of generality, let $t\in (0,\frac{1}{2}]$. 
For any $\epsilon>0$, there exist $\delta>0$ such that $(1-\frac{1}{1+\delta})<\frac{\epsilon}{||\r-\b||_2}$. Let 
$\a'=\frac{\delta2t\a+\delta(1-2t)\b+\r}{1+\delta}$ and $\b'=\frac{\delta\b+\r}{1+\delta}$.
We have $\a',\b'\in \as_{m,n}$ that are both diagonal. Further, $\r=\frac{1}{2}(\a'+\b')$, where $||\r-\a'||_2=||\r-\b'||_2=(1-\frac{1}{1+\delta})||\r-\b||_2<\epsilon$.
Moreover, the two states $\a',\b'$ are linearly independent, otherwise it would imply that $\r=\b=\b'$ and lead to a contradiction. This completes the proof.
$\hfill\square$

\section{Proof of Theorems \ref{ef22}, \ref{mmain} and Corollary \ref{mbn}}
\label{mtp}
We first provide a necessary lemma.
\begin{lemma}
	\label{foc}
	Suppose the order-$n$ matrix $A\ge \mathcal{O}$ has rank one, the order-$n$ matrix $B$ satisfies that $A+B\ge \mathcal{O}$ and $A-B\ge \mathcal{O}$. Then $B$ is linearly dependent with $A$.
\end{lemma}
\begin{proof}
	Let $U$ be the unitary matrix that diagonalizes $A$, i.e., $UAU^\dg=\diag(a,0,\cdots,0)$ for $a>0$. 
	We have 
	\begin{eqnarray}
		\label{nbv}
		\notag
		&&\diag(a,0,\cdots,0)+UBU^\dg\ge \mathcal{O}, \\
		&&\diag(a,0,\cdots,0)-UBU^\dg\ge \mathcal{O}.
	\end{eqnarray}
	Write $UBU^\dg:=[b_{ij}]$. From (\ref{nbv}), we obtain that $b_{ii}=0$ for $i\neq 1$, consequently, $b_{ij}=0$ for any $(i,j)\neq (1,1)$. So $UBU^\dg=\diag(b_{11},0,\cdots,0)$, which is linearly dependent with $UAU^\dg$. This implies that $B$ is also linearly dependent with $A$.
	\end{proof}

\vspace{0.2cm}
{\bf Proof of Theorem \ref{ef22}}

Without loss of generality, we can write $\r=\diag(\l_1,\l_2,\l_3,\l_4)$.
The case of deficient-rank has been proved by Lemma \ref{le:rankrhoA=m} (ii). In the following, we assume that $\r$ has full rank. 

We begin with the only if part.  Suppose $\r$ is an extreme point. Firstly, we have already known that $\r$ is a boundary point,  and hence satisfies condition (I). Next, assume $\l_1>\l_1>\l_3>\l_4$. Let 
\begin{eqnarray}
	\notag
	&&M:=\frac{1}{1+2x}\diag(\l_1+x,\l_2,\l_3+x,\l_4),\\
	&&N:=\frac{1}{1-2x}\diag(\l_1-x,\l_2,\l_3-x,\l_4),
\end{eqnarray}
where $0<x<\min\{\l_1-\l_2,\l_2-\l_3,\l_3-\l_4\}$. It is straightforward to see that the diagonal entries (eigenvalues) of both $M,N$ are in decreasing order and satisfy (\ref{eq:2xnAS}). So $M,N\in\as_{2,2}$. Note that $M,N$ are linearly independent since $x\neq 0$. Hence $\r=\frac{1+2x}{2}M+\frac{1-2x}{2}N$ implies that $\r$ is a non-extreme point and leads to a contradiction. Thus $\r$ must satisfy condition (II). 

We next prove the if part. Suppose $\r$ is a boundary point and at least two 
of $\l_1,\l_2,\l_3,\l_4$ are equal. There are two cases at first, i.e., $\r$ has exactly two or three distinct eigenvalues. 

We begin with the first case that $\r$ has exactly two distinct eigenvalues. Combining with the boundary condition, $\r$ is either $\frac{1}{6}\diag(3,1,1,1)$ or $\frac{1}{8+4\sqrt{2}}\diag(3+2\sqrt{2},3+2\sqrt{2},1,1)$. The state $\frac{1}{6}\diag(3,1,1,1)$ has been proved to be an extreme point by Lemma \ref{le:rankrhoA=m} (i). Suppose the state $\r=\frac{1}{8+4\sqrt{2}}\diag(3+2\sqrt{2},3+2\sqrt{2},1,1)$ is a non-extreme point. Using Theorem \ref{thm:main} (ii), there exist two linearly independent states  $\g=\diag(\g_1,\g_2,\g_3,\g_4)\in \as_{2,2}$ and $\eta=\diag(\eta_1,\eta_2,\eta_3,\eta_4)\in \as_{2,2}$, such that $\r=\frac{1}{2}(\g+\eta)$ and $||\r-\g||_2=||\r-\eta||_2<\frac{1}{10}$.
Let $z_i=\g_i-\l_i=\l_i-\eta_i$ for $i=1,2,3,4$. We can rewrite
\begin{eqnarray}
	\notag
	&&\g=\frac{1}{8+4\sqrt{2}}\diag(3+2\sqrt{2}+z_1,3+2\sqrt{2}+z_2,1+z_3,1+z_4),\\
	&&\eta=\frac{1}{8+4\sqrt{2}}\diag(3+2\sqrt{2}-z_1,3+2\sqrt{2}-z_2,1-z_3,1-z_4),
\end{eqnarray}
where $z_i$ (not all zero) satisfy that $\sum_{i=1}^4 z_i=0$ and $\sqrt{\sum_{i=1}^{4}z_i^2}<\frac{1}{10}$. Further, for both $\g$ and $\eta$, the first two diagonal entries are larger than the last two diagonal entries. If $z_1+z_2\ge 0$, then it follows from Lemma \ref{hjk} (iii) that $\g\succ \frac{1}{8+4\sqrt{2}}\diag(3+2\sqrt{2},3+2\sqrt{2},1,1)$, where the two states have distinct eigenvalues. By Theorem \ref{plpl}, we know that $\g\notin \as_{2,2}$.  Similarly, if $z_1+z_2<0$, then we have $\eta\notin \as_{2,2}$. This contradicts with the assumption. Hence $\r$ is an extreme point.

We next consider the second case that $\r$ has three distinct eigenvalues. There are three subcases: $\l_1>\l_2>\l_3=\l_4$, $\l_1>\l_2=\l_3>\l_4$, and $\l_1=\l_2>\l_3>\l_4$.
Here we only prove the first subcase. The proof of other two cases are omitted since they are completely similar.

Let $\l_1>\l_2>\l_3=\l_4$ and assume that $\r$ is a non-extreme point. 
Firstly, let $\epsilon=\frac{1}{10}\min\{\l_1-\l_2,\l_2-\l_3,\l_3\}$. Using Theorem \ref{thm:main} (ii), there exist two linearly independent states $\a=\diag(\a_1,\a_2,\a_3,\a_4)$, $\b=\diag(\b_1,\b_2,\b_3,\b_4)\in\as_{2,2}$, such that $\r=\frac{1}{2}(\a+\b)$ and 
$||\r-\a||_2=||\r-\b||_2<\epsilon$. Let $x_i=\a_i-\l_i=\l_i-\b_i$ for $i=1,\cdots,4$. So $x_i$ are not all zero. We can rewrite 
\begin{eqnarray}
	\label{abc}
	\notag
	&&\a=\diag(\l_1+x_1,\l_2+x_2,\l_3+x_3,\l_3+x_4),\\
	&&\b=\diag(\l_1-x_1,\l_2-x_2,\l_3-x_3,\l_3-x_4),
\end{eqnarray}
where $\sum_{i=1}^4x_i=0$ and $\sqrt{\sum_{i=1}^{4}x_i^2}<\epsilon$. So $|x_1|,\cdots,|x_4|<\epsilon$. Next, let 
\begin{eqnarray}
	\label{yyy}
	\notag
	&&\a'=\diag(\l_1+x_1,\l_2+x_2,\l_3+\frac{x_3+x_4}{2},\l_3+\frac{x_3+x_4}{2}),\\
	&&\b'=\diag(\l_1-x_1,\l_2-x_2,\l_3-\frac{x_3+x_4}{2},\l_3-\frac{x_3+x_4}{2}).
\end{eqnarray}
From the expression of $\epsilon$, one can directly verify that the diagonal entries of both $\a',\b'$ are in non-increasing order. On the other hand,
using Lemma \ref{hjk} (ii), we have $\a\succ \a'$ and $\b\succ \b'$, respectively. So $\a',\b'\in \as_{2,2}$ by Theorem \ref{maj}. 
According to the criterion
(\ref{drt}), we obtain that
\begin{eqnarray}
	\label{iui}
	\notag
	\bma 2(\l_3+\frac{x_3+x_4}{2})&&\l_3+\frac{x_3+x_4}{2}-(\l_1+x_1)\\\l_3+\frac{x_3+x_4}{2}-(\l_1+x_1)&&2(\l_2+x_2)\ema \ge \mathcal{O},\\
	\bma 2(\l_3-\frac{x_3+x_4}{2})&&\l_3-\frac{x_3+x_4}{2}-(\l_1-x_1)\\\l_3-\frac{x_3+x_4}{2}-(\l_1-x_1)&&2(\l_2-x_2)\ema \ge \mathcal{O}.
\end{eqnarray}
Recalling that $\r$ is a boundary point and thus $\bma 2\l_4&&\l_3-\l_1\\\l_3-\l_1&&2\l_2\ema$ has rank one. From (\ref{iui}) and using Lemma \ref{foc}, the two matrices 
$\bma 2\l_3&&\l_3-\l_1\\\l_3-\l_1&&2\l_2\ema$ and $\bma 2\cdot \frac{x_3+x_4}{2}&&\frac{x_3+x_4}{2}-x_1\\\frac{x_3+x_4}{2}-x_1&&2x_2\ema$ are linearly dependent. Consequently, the two vectors $(\l_1,\l_2,\l_3,\l_3)$ and $(x_1,x_2,\frac{x_3+x_4}{2},\frac{x_3+x_4}{2})$ are linearly dependent.
Due to the fact that $\sum_{i=1}^4 x_i=0$, we have $x_1=x_2=\frac{x_3+x_4}{2}=0$. Taking back to (\ref{abc}), we have $\a=\diag(\l_1,\l_2,\l_3+x_3,\l_3-x_3)$ with $x_3\neq 0$. It follows that $\a\succ \r$ but does not have the same eigenvalues as $\r$. However, since $\r$ is a boundary point, using Theorem \ref{plpl}, $\a\notin \as_{2,2}$. This is a contradiction. Hence the assumption does not hold and $\r$ is an extreme point.
$\hfill\square$

\vspace{0.2cm}
{\bf Proof of Theorem \ref{mmain}}

Without loss of generality, let $\rho=\diag(\lambda_1,\cdots,\lambda_{2n})$. 
The case of deficient-rank has been proved by Lemma \ref{le:rankrhoA=m} (ii). In the following, we assume that $\r$ has full rank.

We begin with the only if part. Suppose $\r$ is an extreme point. So $\r$ is a boundary point, and condition (I) should be satisfied. 
	Next, assume that condition (II) is violated. This implies that $n>2$ and $\lambda_1>\l_k>\lambda_{2n-2}$ for at least one $2\le k\le 2n-3$.
	Let $t$ satisfies that $\lambda_{1}>\lambda_k\pm t>\lambda_{2n-2}$ and
	\begin{eqnarray}
		\notag
		&&\a=\frac{1}{1+t}\diag(\lambda_1,\cdots,\lambda_{k-1},\lambda_k+t,\lambda_{k+1},\cdots,\lambda_{2n}),\\
		&&\b=\frac{1}{1-t}\diag(\lambda_1,\cdots,\lambda_{k-1},\lambda_k-t,\lambda_{k+1},\cdots,\lambda_{2n}).
	\end{eqnarray}
	One can verify that $\a,\b\in \as_{2,n}$ are linearly independent. Hence, $\rho=\frac{1+t}{2}\a+\frac{1-t}{2}\b$ implies that it is non-extreme. This is a contradiction. Hence condition (II) must hold.
		Finally, suppose conditions (I) and (II) hold but (III) is violated.  It follows that 
$
\lambda_1=\cdots=\lambda_k>\lambda_{k+1}=\cdots=\lambda_{2n-2}>\lambda_{2n-1}>\lambda_{2n},
$
where $1\le k\le 2n-3$.	Let
$0<x<\min\{\lambda_{1}-\lambda_{2n-2},\lambda_{2n-2}-\lambda_{2n-1},\lambda_{2n-1}-\lambda_{2n},\frac{1}{k+1}\}$,
and
\begin{eqnarray}
	\notag
	&&\g=\frac{1}{1+kx+x}\diag(\l_1+x,\cdots,\l_k+x,\l_{k+1},\cdots,\l_{2n-2},\l_{2n-1}+x,\l_{2n}),\\
	&&\eta=\frac{1}{1-kx-x}\diag(\l_1-x,\cdots,\l_k-x,\l_{k+1},\cdots,\l_{2n-2},\l_{2n-1}-x,\l_{2n}).
	\end{eqnarray}
From the range of $x$, we have the largest eigenvalue of $\g$ is $\l_1+x$, while the three smallest eigenvalues are $\l_{2n-2}\ge\l_{2n-1}+x\ge\l_{2n}$. According to the criterion (\ref{eq:2xnAS}), we have $\g\in\as_{2,n}$. Similarly, we have $\eta\in\as_{2,n}$. Consequently, $\r=\frac{1+kx+x}{2}\g+\frac{1-kx-x}{2}\eta$ implies $\r$ is a non-extreme point. This is a contradiction. Hence condition (III) must hold. We have proved the only if part.
	
	We next prove the if part. Let $\rho$
	satisfy condition (I)-(III). Suppose $\rho$ is a non-extreme point. 
	It follows from condition (II) that  
	\begin{eqnarray}
		\label{lkj}
	\lambda_1=\cdots=\lambda_k\ge\lambda_{k+1}=\cdots=\lambda_{2n-2}\ge\lambda_{2n-1}\ge \lambda_{2n}
	\end{eqnarray}
	for $1\le k\le 2n-3$. Using Theorem \ref{thm:main} (ii), there exist two linearly independent diagonal states $M,N\in\as_{2,n}$ such that $\r=\frac{1}{2}(M+N)$. Let $\mathcal{J}_1=\{1,\cdots,k\}$, $\mathcal{J}_2=\{k+1,\cdots,2n-2\}$ and $
	\cK=\{a,b,2n-1,2n\}$,
	where $a\in J_1$ and $b\in J_2$. We have 
	$\r_{\cK}=\frac{1}{2}(M_{\cK}+N_{\cK})$.
Using Theorem \ref{le:mxn=AS}, we obtain that 
	$\r_{\cK}, M_{\cK}, N_{\cK}\in \as_{2,2}$. 
	Further, since $\r$ satisfies conditions (I), (III) and (\ref{lkj}), it follows from the definition of $\cK$ and Theorem \ref{ef22} that $\r_{\cK}$ is an extreme point of $\as_{2,2}$.
 This implies that $M_{\cK}$ and $N_{\cK}$ are linearly dependent with $\r_{\cK}$.
	Let $a$, $b$ go through $\mathcal{J}_1$ and 
	$\mathcal{J}_2$ respectively,  the two states $M_{\cK}$ and $N_{\cK}$ are always linearly dependent. We conclude that $M,N$ are linearly dependent, which leads to a contradiction. Hence $\r$ is an extreme point.
	This completes the proof.
$\hfill\square$

\vspace{0.2cm}
{\bf Proof of Corollary \ref{mbn}}

(i) The if part can be verified directly. We next prove the only if part. Suppose the extreme point $\r$ resides on the maximal ball. Using Theorem \ref{ef22}, we know that $\r$ is a boundary point of $\as_{2,2}$ and at least two of $\lambda_1,\lambda_2,\lambda_3,\lambda_4$ are equal. The case of deficient rank has been proved in Lemma \ref{le:rankrhoA=m} (ii). We only need to assume that $\r$ has full rank. Let $k_1:=\frac{\lambda_2}{\lambda_4}\ge 1$. If $\lambda_3=\lambda_4$, the eigenvalues vector of $\r$ can be written as  $\frac{1}{3+2\sqrt{k_1}+k_1}(1+2\sqrt{k_1},k_1,1,1)$. A calculation gives that $\tr(\r^2)-\frac{1}{3}=\frac{2 \left(\sqrt{k_1}-1\right)^2 k_1}{3\left(k_1+2 \sqrt{k_1}+3\right)^2}\ge 0$. Hence the inequality is satisfied only if $k_1=1$ and it follows that $\l_2=\l_3=\l_4$. Similarly, if $\l_2=\l_3$ then we also have $k_1=1$. If $\l_1=\l_2$, then 
	$\tr(\r^2)-\frac{1}{3}=\frac{2 \left(\sqrt{k_1}+1\right)^2 k_1}{3\left(k_1-2 \sqrt{k_1}+3\right)^2}>0$. We conclude that $\r$ resides on the maximal ball only if $\l(\r)=(\frac{1}{2},\frac{1}{6},\frac{1}{6},\frac{1}{6})$. 
	
	(ii) The if part can be verified directly. We next prove the only if part. It suffices to prove that any full-rank extreme point $\r$ resides outside the maximal ball. Let $k_2:=\frac{\lambda_{2n-2}}{\lambda_{2n}}\ge 1$. According to condition (III) in Theorem \ref{mmain}, there are three cases based on $\lambda_1,\lambda_{2n-2},\lambda_{2n-1},\lambda_{2n}$. Suppose $\lambda_{1}=\lambda_{2n-2}$. Then $\l(\r)=\frac{1}{(2n-1)k_2-2\sqrt{k_2}+1}(k_2,\cdots,k_2,k_2-2\sqrt{k_2},1)$. A computation in Mathematica yields
	\begin{eqnarray}
		\tr(\r^2)-\frac{1}{2n-1}=\frac{2 \left(k_2 (2 n-3)+2 \sqrt{k_2}+n-1\right)}{(2 n-1) \left(-2 k_2 n+k_2+2 \sqrt{k_2}-1\right)^2}>0,
	\end{eqnarray}
	which implies that $\r$ resides outside the maximal ball. For the other two cases where $\lambda_{2n-2}=\lambda_{2n-1}$ or $\lambda_{2n-1}=\lambda_{2n}$, we apply a similar way and it turns out that $\tr(\r^2)>\frac{1}{2n-1}$ holds. We conclude that the claim holds.
	$\hfill\square$

\section{Proof of results in Section \ref{sec:3x3}}
\label{profL1+L2}
Before proving the main results, we require the following two lemmas.
\begin{lemma}
		\label{pf3e}
		(i) Suppose $x\in \cZ_9^+$ such that $L_1(x)\ge \mathcal{O}$. Then 
		\begin{eqnarray}
			\label{331}
			x_1&&\le x_8+2\sqrt{x_7x_9}\\
			\label{3331}
			&&\le x_7+x_8+x_9.
		\end{eqnarray}
		In particular, if (\ref{331}) is saturated, then $x_2=\cdots=x_6$.
		
		(ii) Suppose $x\in \cZ_9^+$ such that $L_2(x)\ge \mathcal{O}$. Then 
		\begin{eqnarray}
			\label{r331}
			x_1&&\le x_8+2\sqrt{x_6x_9}\\
			\label{r3331}
			&&\le x_6+x_8+x_9.
		\end{eqnarray}
		In particular, if (\ref{r331}) is saturated, then $x_2=\cdots=x_7$.

		(iii) Suppose $x\in \cZ_9^+$ such that $L_1(x)\ge \mathcal{O}$ and $l_1(x)=0$. Then 
		\begin{eqnarray}
			\label{det89}
			&&(x_3-x_5)^2-4x_4x_7+(x_3-x_5)(x_2-x_6)+2x_4(x_1-x_8)\le 0,\\
			\label{det69}
			&&(x_3-x_5)^2-4x_4x_7+(x_3-x_5)(x_1-x_8)+2x_7(x_2-x_6)\le 0,\\
			\label{det45}
			&&(x_1-x_8)^2-4x_7x_9+(x_1-x_8)(x_2-x_6)+2x_9(x_3-x_5)\ge 0,\\
			\label{det444}
			&&(x_1-x_8)^2-4x_7x_9+(x_1-x_8)(x_3-x_5)+2x_7(x_2-x_6)\ge 0,\\
			\label{zlt}
			&&(x_2-x_6)^2-4x_4x_9+(x_2-x_6)(x_1-x_8)+2x_9(x_3-x_5)\le 0,\\
			\label{zllt}
			&&(x_2-x_6)^2-4x_4x_9+(x_2-x_6)(x_3-x_5)+2x_4(x_1-x_8)\ge 0.
		\end{eqnarray} 
		
		(iv) Suppose $x\in \cZ_9^+$ such that $L_2(x)\ge \mathcal{O}$ and $l_2(x)=0$. Then 
		\begin{eqnarray}
			\label{det289}
			&&(x_3-x_5)^2-4x_4x_6+(x_3-x_5)(x_2-x_7)+2x_4(x_1-x_8)\le 0,\\
			\label{det269}
			&&(x_3-x_5)^2-4x_4x_6+(x_3-x_5)(x_1-x_8)+2x_6(x_2-x_7)\le 0,\\
			\label{det245}
			&&(x_1-x_8)^2-4x_6x_9+(x_1-x_8)(x_2-x_7)+2x_9(x_3-x_5)\ge 0,\\
			\label{det2444}
			&&(x_1-x_8)^2-4x_6x_9+(x_1-x_8)(x_3-x_5)+2x_6(x_2-x_7)\ge 0,\\
			&&(x_2-x_7)^2-4x_4x_9+(x_2-x_7)(x_1-x_8)+2x_9(x_3-x_5)\le 0,\\
			&&(x_2-x_7)^2-4x_4x_9+(x_2-x_7)(x_3-x_5)+2x_4(x_1-x_8)\ge 0.
		\end{eqnarray} 
	\end{lemma}
	\begin{proof}
	(i) The inequality (\ref{331}) follows from Lemma \ref{le:rankrhoA=m} (i), and (\ref{3331}) follows directly from (\ref{331}). Suppose $x_1=x_8+2\sqrt{x_7x_9}$, a direct calculation gives
\begin{eqnarray}
	l_1(x)=2(x_8-x_1)(x_6-x_2)(x_5-x_3)
	-2(x_6-x_2)^2x_7-2(x_5-x_3)^2x_9\le 0.
\end{eqnarray}
Hence $L_1(x)\ge \mathcal{O}$ only if $l_1(x)=0$, which implies $x_2-x_6=0$. 

(ii) The proof is similar to (i).

(iii) We first prove (\ref{det89}) and (\ref{det45}), followed by proving (\ref{det69}) and (\ref{det444}). The proofs of (\ref{zlt}) and (\ref{zllt}) are similar to (\ref{det89}) and (\ref{det45}) respectively, hence we will skip them here.

Proof of (\ref{det89}): Since $L_1(x)\ge \mathcal{O}$, we know that $4x_4x_9-(x_2-x_6)^2\ge 0$. Suppose $4x_4x_9-(x_2-x_6)^2=0$, equivalently, $x_2=x_6+2\sqrt{x_4x_9}$. Combining with (\ref{331}), we have 
$
	x_2\le x_1\le x_8+2\sqrt{x_7x_9}\le x_6+2\sqrt{x_4x_9},
$
where all the inequalities are saturated. This implies that $x_1=x_2$ and $x_6=x_8$. According to (i), we obtain that $x_1=\cdots=x_8$ and (\ref{det89}) holds. On the other hand, suppose $4x_4x_9-(x_2-x_6)^2>0$. We also have $4x_4x_7-(x_3-x_5)^2>0$. By calculation,
$
	l_1(x)=-2x_4 \cdot r_1 \cdot r_2,
$
where
\begin{eqnarray}
	\notag
	&&r_1=x_8-x_1-\frac{(x_2-x_6)(x_3-x_5)+\sqrt{(4x_4x_7-(x_3-x_5)^2)(4x_4x_9-(x_2-x_6)^2)}}{2x_4},\\
	\notag
	&&r_2=x_8-x_1-\frac{(x_2-x_6)(x_3-x_5)-\sqrt{(4x_4x_7-(x_3-x_5)^2)(4x_4x_9-(x_2-x_6)^2)}}{2x_4}.
\end{eqnarray}
One can verify that $r_1<0$. Thus the hypothesis that $l_1(x)=0$ implies  $r_2=0$.
Considering the leftside of (\ref{det89}), we have
\begin{eqnarray}
	\notag
	&&(x_3-x_5)^2-4x_4x_7+(x_3-x_5)(x_2-x_6)+2x_4(x_1-x_8)\\
	\notag
	=&&-x_4(1-\sqrt{\frac{4x_4x_7-(x_3-x_5)^2}{4x_4x_9-(x_2-x_6)^2}})r_1
	-x_4(1+\sqrt{\frac{4x_4x_7-(x_3-x_5)^2}{4x_4x_9-(x_2-x_6)^2}})r_2\\
	=&&-x_4(1-\sqrt{\frac{4x_4x_7-(x_3-x_5)^2}{4x_4x_9-(x_2-x_6)^2}})r_1.
\end{eqnarray}
Thus (\ref{det89}) holds since $r_1<0$ and  $1-\sqrt{\frac{4x_4x_7-(x_3-x_5)^2}{4x_4x_9-(x_2-x_6)^2}}\le 0$.

Proof of (\ref{det45}): It is straightforward to see that the inequality holds if $(x_1-x_8)^2-4x_7x_9=0$. Suppose $(x_1-x_8)^2-4x_7x_9<0$, we also have
$(x_2-x_6)^2-4x_4x_9<0$. A calculation gives
$
	l_1(x)=-2 x_9\cdot r_3\cdot r_4,
$
where
\begin{eqnarray}
	\notag
	&&r_3=x_5-x_3-\frac{(x_2-x_6)(x_1-x_8)+\sqrt{(4x_7x_9-(x_1-x_8)^2)(4x_4x_9-(x_2-x_6)^2)}}{2x_9},\\
	\notag
	&&r_4=x_5-x_3-\frac{(x_2-x_6)(x_1-x_8)-\sqrt{(4x_7x_9-(x_1-x_8)^2)(4x_4x_9-(x_2-x_6)^2)}}{2x_9}.
\end{eqnarray}
One can verify that $r_3<0$. Thus the hypothesis that $l_1(x)=0$ implies  $r_4=0$.
Considering the leftside of (\ref{det45}), we have
\begin{eqnarray}
	\notag
	&&(x_1-x_8)^2-4x_7x_9+(x_1-x_8)(x_2-x_6)+2x_9(x_3-x_5)\\
	\notag
	=&&x_9(-1+\sqrt{\frac{4x_7x_9-(x_1-x_8)^2}{4x_4x_9-(x_2-x_6)^2}})r_3
	+x_9(-1-\sqrt{\frac{4x_7x_9-(x_1-x_8)^2}{4x_4x_9-(x_2-x_6)^2}})r_4\\
	=&&x_9(-1+\sqrt{\frac{4x_7x_9-(x_1-x_8)^2}{4x_4x_9-(x_2-x_6)^2}})r_3.
\end{eqnarray}
Thus (\ref{det45}) holds since $r_3<0$ and  $-1+\sqrt{\frac{4x_7x_9-(x_1-x_8)^2}{4x_4x_9-(x_2-x_6)^2}}\le 0$.

Proof of (\ref{det69}): By calculation, the difference between the leftsides of (\ref{det69}) and (\ref{det89}) is
\begin{eqnarray}
(x_3-x_5-2x_4)(x_1-x_8)+(2x_7-x_3+x_5)(x_2-x_6)\le 2(x_1-x_8)(x_7-x_4)\le 0,
\end{eqnarray}
where the first inequality follows from (\ref{3331}). Thus (\ref{det69}) holds according to (\ref{det89}).

Proof of (\ref{det444}): A direct calculation gives the difference between the leftsides of (\ref{det444}) and (\ref{det45}) is
\begin{eqnarray}
(x_1-x_8-2x_9)(x_3-x_5)+(2x_7-x_1+x_8)(x_2-x_6)\ge 2(x_3-x_5)(x_7-x_9)\ge 0,
\end{eqnarray}
where the first inequality follows from (\ref{3331}). Thus (\ref{det444}) holds according to (\ref{det45}).

(iv) The proof is similar to (iii).
	\end{proof}

\begin{lemma}
	\label{ppp}
	Define the sets
	\begin{eqnarray}
		\label{cD1}
		&&\cD_1:=\{x\in \cZ_9^+|L_1(x)\ge \mathcal{O}, L_2(x)\ge \mathcal{O}, l_1(x)=0\},\\
		&&\cD_2:=\{x\in \cZ_9^+|L_1(x)\ge \mathcal{O}, L_2(x)\ge \mathcal{O}, l_2(x)=0\}.
	\end{eqnarray}
	
	(i) Given a vector $x\in \cD_1$ and a pair of integers $(i,j)$ satisfies that $1\le i<j\le 9$ and $(i,j)\neq (6,7)$. Suppose there exists $t>0$ such that the vector $(x_1,\cdots,x_{i-1},x_i+t,x_{i+1}\cdots,x_{j-1},x_j-t,x_{j+1},\cdots, x_9)\in \cZ_9^+$. Then there exists $\delta\in (0,t)$ such that $l_1(y)<0$, where 
	$y:=(x_1,\cdots,x_{i-1},x_i+\delta,x_{i+1}\cdots,x_{j-1},x_j-\delta,x_{j+1},\cdots, x_9)\in \cZ_9^+$.
	
	(ii) Given a vector $x\in \cD_1$ with $x_6>x_7$. Suppose there exists $t>0$ such that $(x_1,\cdots,x_{5},x_6+t,x_{7}-t,x_{8},x_9)\in \cZ_9^+$. Then there exists 
	$\delta\in (0,t)$ such that 
$l_1((x_1,\cdots,x_{5},x_6+\delta,x_{7}-\delta,x_{8},x_9))<0$.

	(iii) Given a vector $x\in \cD_2$ and a pair of integers $(i,j)$ satisfies that $1\le i<j\le 9$ and $(i,j)\neq (6,7)$. Suppose there exists $t>0$ such that the vector $(x_1,\cdots,x_{i-1},x_i+t,x_{i+1}\cdots,x_{j-1},x_j-t,x_{j+1},\cdots, x_9)\in \cZ_9^+$. Then there exists $\delta\in (0,t)$ such that $l_2(y)<0$, where 
	$y:=(x_1,\cdots,x_{i-1},x_i+\delta,x_{i+1}\cdots,x_{j-1},x_j-\delta,x_{j+1},\cdots, x_9)\in \cZ_9^+$.

	(iv)  Given a vector $x\in \cD_2$ with $x_6>x_7$. Suppose there exists $t>0$ such that $(x_1,\cdots,x_{5},x_6+t,x_{7}-t,x_{8},x_9)\in \cZ_9^+$. Then there exists 
	$\delta\in (0,t)$ such that 
$l_2((x_1,\cdots,x_{5},x_6+\delta,x_{7}-\delta,x_{8},x_9))<0$.
\end{lemma}
\begin{proof}
(i) There are 35 possible cases of $(i,j)$ to consider, which we divide into eight kinds of cases as follow:
\begin{eqnarray}
	\notag
	&&A1=\{(4,7),(4,9),(7,9)\},\\
	\notag
	&&A2=\{(1,8),(2,6),(3,5)\},\\
	\notag
	&&A3=\{(1,2),(1,3),(2,3),(5,6),(5,8),(6,8)\},\\
	\notag
	&&A4=\{(2,4),(3,4),(1,7),(3,7),(1,9),(2,9)\},\\
	\notag
	&&A5=\{(1,4),(2,9),(3,7)\},\\
	\notag
	&&A6=\{(1,5),(1,6),(2,5),(2,8),(3,6),(3,8)\},\\ \notag
	&&A7=\{(4,5),(4,6),(5,7),(7,8),(6,9),(8,9)\},\\
	\notag
	&&A8=\{(4,8),(5,9)\}.
\end{eqnarray}
We point out that the proofs for different cases within the same category are essentially identical, requiring no specific techniques but rather straightforward calculations. Therefore, we present a detailed proof of one case in each category, with the proofs for the remaining cases being easily derivable.

Proof of A1:
Suppose $(i,j)=(4,7)$ and there exists a $t>0$ such that 
$(x_1,x_2,x_3,x_4+t,x_5,x_6,x_7-t,x_8,x_9)\in \cZ_9^+$. Thus for any $\delta\in (0,t)$, $y=(x_1,x_2,x_3,x_4+\delta,x_5,x_6,x_7-\delta,x_8,x_9)\in \cZ_9^+$.
Consequently, a calculation gives 
\begin{eqnarray}
l_1(y)=l_1(y)-l_1(x)
	=-2\delta[(x_1-x_8)^2-(x_2-x_6)^2+4(x_4-x_7)x_9]-8x_9\delta^2<0.
\end{eqnarray}
Thus the claim holds. 

Proof of A2:
Suppose $(i,j)=(1,8)$. For any $\delta\in (0,t)$, a calculation gives
\begin{eqnarray}
	l_1(y)
	=-4\delta[(x_2-x_6)(x_3-x_5)+4x_4(x_1-x_8)]-8x_4\delta^2<0.
\end{eqnarray}

Proof of A3:
Suppose $(i,j)=(1,2)$. We have $x_2>x_3$. For any $\delta\in (0,t)$, a  calculation gives
\begin{eqnarray}
	\label{dd12}
	\notag
	l_1(y)
	=&&-2\delta[(x_6-x_2)(2x_7+x_5-x_3)+(x_1-x_8)(2x_4+x_5-x_3)]
	-2\delta^2(x_4+x_7+x_5-x_3)\\
	\le&&-2\delta(x_1-x_8)(2x_4-2x_7)-2\delta^2(x_4+x_7+x_5-x_3)\le 0, 
\end{eqnarray}
where the two inequalities follow from (\ref{3331}). Further, since 
$
	x_3<x_2\le x_1,
$
we obtain that $l_1(y)$ is strictly less than zero. 

Proof of A4:
Suppose $(i,j)=(2,4)$. For any $\delta\in (0,t)$, a calculation gives
\begin{eqnarray}
	\label{dd24}
	l_1(y)
	=-2\delta[4x_7x_9-(x_1-x_8)^2+(x_1-x_8)(x_3-x_5)+2x_7(x_2-x_6)]-2x_7\delta^2<0,
\end{eqnarray}
where the inequality follows from (\ref{331}). 

Proof of A5:
Suppose $(i,j)=(1,4)$. For any $\delta\in (0,t)$, a  calculation gives
	\begin{eqnarray}
		\label{dd14}
		l_1(y)
		=2\delta[(x_1-x_8)^2-4x_7x_9-(x_2-x_6)(x_3-x_5)
		-2x_4(x_1-x_8)+\delta^2-\delta(x_4+2x_8-2x_1)].
\end{eqnarray}
Since $(x_1-x_8)^2-4x_7x_9-(x_2-x_6)(x_3-x_5)-2x_4(x_1-x_8)<0$, $\delta$ can be small enough such that  $l_1(y)<0$.

Proof of A6:
Suppose $(i,j)=(1,5)$. For any $\delta\in (0,t)$, a calculation gives
\begin{eqnarray}
	\label{dd15}
	l_1(y)=2\delta[(x_5-x_3)(2x_9+x_2-x_6)+(x_8-x_1)(2x_4+x_2-x_6)]
	-\delta^2(2x_9+2x_4+2x_2-2x_6)<0.
\end{eqnarray}

Proof of A7:
Suppose $(i,j)=(4,5)$. For any $\delta\in (0,t)$, a calculation gives
\begin{eqnarray}
l_1(x_{4,5}(\delta))=&&-2\delta[(x_1-x_8)^2-4x_7x_9+(x_1-x_8)(x_2-x_6)+2x_9(x_3-x_5)]-2x_9\delta^2<0,
\end{eqnarray}
which follows from (\ref{det45}). The remaining five cases in A7 can be verified in a similar manner by applying the other five inequalities
in Lemma \ref{pf3e} (iii) accordingly.

Proof of A8:
Suppose $(i,j)=(4,8)$. For any $\delta\in (0,t)$, a calculation gives
	\begin{eqnarray}
		l_1(y)
	=-2\delta[(x_1-x_8)^2-4x_7x_9+(x_3-x_5)(x_2-x_6)
	+2x_4(x_1-x_8)+\delta^2+2(x_1-x_8)\delta]
		-2x_4\delta^2.
	\end{eqnarray}
Recalling from (\ref{3331}), we have
\begin{eqnarray}
	\label{zz1}
	\notag
	&&(x_1-x_8)^2-4x_7x_9+(x_3-x_5)(x_2-x_6)+2x_4(x_1-x_8)\\
	\notag
	-&&[(x_1-x_8)^2-4x_7x_9+(x_1-x_8)(x_3-x_5)+2x_7(x_2-x_6)]\\
=&&(x_3-x_5-2x_7)(x_2-x_6)+(x_1-x_8)(2x_4+x_5-x_3)
	\ge(x_2-x_6)(2x_4-2x_7)\ge 0.
\end{eqnarray}
Combining (\ref{zz1}) with (\ref{det444}), we obtain that $(x_1-x_8)^2-4x_7x_9+(x_3-x_5)(x_2-x_6)+2x_4(x_1-x_8)\ge 0$. It then follows that $l_1(y)<0$. The other case in A8 can be proved similarly by using the inequality (\ref{det69}).

In conclusion, the claim holds through the aforementioned proof.

(ii) For any $\delta\in (0,t)$, a calculation gives
\begin{eqnarray}
		\label{lol}
		\notag
		&&l_1((x_1,\cdots,x_{5},x_6+\delta,x_{7}-\delta,x_{8},x_9))\\
	=&&2\delta[(x_2-x_6)(x_2-x_6+2x_7)+(x_3-x_5)(x_1-x_8)-4x_4x_9
		-(2x_2-2x_6+x_7)\delta+\delta^2].
	\end{eqnarray}
Since $x\in \cD_1$, we have $l_2(x)-l_1(x)\ge 0$. Another computation yields
\begin{eqnarray}
	(x_2-x_6)(x_2-x_6+2x_7)+(x_3-x_5)(x_1-x_8)-4x_4x_9+l_2(x)-l_1(x)
	=-(2x_2-x_6)(x_6-x_7)<0.
\end{eqnarray}
This implies that $(x_2-x_6)(x_2-x_6+2x_7)+(x_3-x_5)(x_1-x_8)-4x_4x_9<0$. 
It follows that $\delta$ can be small enough such that $l_1((x_1,\cdots,x_{5},x_6+\delta,x_{7}-\delta,x_{8},x_9))<0$.

(iii) The proof is identical to that of (i). In other words, each case of $(i,j)$ can be handled through direct calculations and the application of the results in Lemma \ref{pf3e} (iii) and (iv).

(iv) The proof is completely similar to that of (ii) and can be derived through direct computations.
\end{proof}

\vspace{0.2cm}
{\bf Proof of Theorem \ref{3bb}}

Firstly, by Lemma \ref{le:rankrhoA=m} (ii) and Theorem \ref{maj}, we obtain that the claim holds if $\r$ has deficient rank. In the following, assume that $\r$ has full rank. For ease of notation, we denote the non-increasing ordered eigenvalue vectors of $\r$ and $\sigma$ as $\l=(\l_1,\cdots,\l_9)$ and $\mu=(\mu_1,\cdots,\mu_9)$, respectively. Define  $\tilde{x}_k:=\sum_{i=1}^k x_i$. From the hypothesis, we have 
\begin{eqnarray}
	\label{3ma4}
	&&\tilde{\mu}_i\ge \tilde{\l}_i, i=1,\cdots,8, \\
	\label{3ma3}
	&&\tilde{\mu}_9=\tilde{\l}_9=1,
\end{eqnarray}
where at least one inequality in (\ref{3ma4}) is strict. 
Define 
\begin{eqnarray}
	&&a:=\min\{i|\tilde{\mu}_i>\tilde{\l}_i\},\\
	&&b:=\min\{j|j>a,\tilde{\mu}_j=\tilde{\l}_j\}.
\end{eqnarray}
It immediates that $1\le a\le 8$ and $\mu_a>\l_a$. If $a>1$, then 
\begin{eqnarray}
	\label{ida-1}
	\mu_i=\l_i, i=1,\cdots,a-1,
\end{eqnarray}
which implies that
\begin{eqnarray}
	\label{aanda-1}
	\l_{a-1}=\mu_{a-1}\ge \mu_a>\l_{a}.
\end{eqnarray}
On the other hand, from the definition of $b$, we know that $a<b\le 9$, and
\begin{eqnarray}
	\label{a+1b-1}
	\tilde{\mu}_i>\tilde{\l}_i, i=a+1,\cdots,b-1.
\end{eqnarray}
In particular, if $b<9$, then 
\begin{eqnarray}
	\label{bandb+1}
	\l_b=\tilde{\l}_b-\tilde{\l}_{b-1}>\tilde{\mu}_b-\tilde{\mu}_{b-1}=\mu_{b}\ge \mu_{b+1}=\tilde{\mu}_{b+1}-\tilde{\mu}_{b}\ge  \tilde{\l}_{b+1}-\tilde{\l}_{b}=\l_{b+1}.
\end{eqnarray}

Recalling that $\r$ is a boundary point of $\app_{3,3}$, according to Lemma \ref{le:L1+L21}, there are two cases: (I) $l_1(\l)=0, l_2(\l)\ge 0$, (II) $l_1(\l)\ge 0, l_2(\l)=0$. Here we prove case (I), it turns out the proof of case (II) is similar.

(I) Suppose $l_1(\l)=0, l_2(\l)\ge 0$. Hence $\l\in \cD_1$ defined in (\ref{cD1}). There are two subcases for $a,b$ to consider:
(I1) $(a,b)\neq (6,7)$, (I2) $(a,b)=(6,7)$.

(I1) Suppose $(a,b)\neq (6,7)$. 
Let 
\begin{eqnarray}
	\label{rot}
	0<t_1<\min\{\tilde{\mu}_a-\tilde{\l}_a,\tilde{\mu}_{a+1}-\tilde{\l}_{a+1},\cdots,\tilde{\mu}_{b-1}-\tilde{\l}_{b-1},\l_{a-1}-\l_a,\l_{b}-\l_{b+1}\},
\end{eqnarray}
where the element $\l_{a-1}-\l_a$ is considered nonexistent if $a=1$, the element $\l_{b}-\l_{b+1}$ is replaced by $\l_9$ if $b=9$.
The positivity of $t_1$ is guaranteed by inequalities (\ref{aanda-1}) -(\ref{bandb+1}). Consequently, $(\l_1,\cdots,\l_{a-1},\l_a+t_1,\l_{a+1}\cdots,\l_{b-1},\l_b-t_1,\l_{b+1},\cdots, \l_9)\in \cZ_9^+$. Using Lemma \ref{ppp} (i), there exists $\delta_1\in (0,t_1)$ such that 
$
l_1(\eta)<0,
$
where $\eta:=(\l_1,\cdots,\l_{a-1},\l_a+\delta_1,\l_{a+1}\cdots,\l_{b-1},\l_b-\delta_1,\l_{b+1},\cdots, \l_9)\in \cZ_9^+$. From the expression of $\eta$ and (\ref{ida-1}), (\ref{rot}), we have
\begin{eqnarray}
	\notag
	\label{yc4}	
	&&\tilde{\mu}_i=\tilde{\l}_i=\tilde{\eta}_i, i=1,\cdots,a-1,\\
	\notag
	\label{yc1}
	&&\tilde{\mu}_i>\tilde{\l}_i+t_1>\tilde{\l}_i+\delta_1=\tilde{\eta}_i, i=a,\cdots,b-1,\\
	\label{yc2}
	&&\tilde{\mu}_i\ge\tilde{\l}_i=\tilde{\eta}_i, i=b,\cdots,9.
\end{eqnarray}
which implies that $\mu\succ \eta$. If $\sigma\in \app_{3,3}$, using Theorem \ref{maj}, $\diag(\eta)\in \app_{3,3}$, and so $L_1(\eta)\ge \mathcal{O}$. However, this contradicts with $l_1(\eta)<0$. Hence $\sigma\notin \app_{3,3}$
and the claim holds.

(I2) Suppose $(a,b)=(6,7)$. It follows from (\ref{aanda-1}) and (\ref{bandb+1}) that $\l_5>\l_6$ and $\l_7>\l_8$. Let
$
0<t_2<\min\{\l_5-\l_6,\l_7-\l_8,\tilde{\mu}_6-\tilde{\l}_6\}.
$
Thus $\omega:=(\l_1,\cdots,\l_5,\l_6+t_2,\l_7-t_2,\l_8,\l_9)\in \cZ_9^+$.
From (\ref{3ma4}), (\ref{3ma3}) and the expression of $\omega$, we have $\mu\succ \omega$.

Suppose $\l_6=\l_7$. If $\sigma \in\app_{3,3}$, then by Theorem \ref{maj}, $\diag(\omega)\in \app_{3,3}$ and hence $l_1(\omega)\ge 0$, $l_2(\omega)\ge 0$. However, a direct calculation gives
$
l_1(\omega)+l_2(\omega)=4t_2^2(-2\l_2+\l_6)<0,
$
which leads to a contradiction. Hence $\sigma\notin \app_{3,3}$.

On the other hand, suppose $\l_6>\l_7$. Using Lemma \ref{ppp} (ii), there exist $\delta_2\in (0,t_2)$ such that  $\kappa:=(\l_1,\cdots,\l_5,\l_6+\delta_2,\l_7-\delta_2,\l_8,\l_9)\in \cZ_9^+$, and
$
l_1(\kappa)<0.
$
It follows that $\omega \succ \kappa$ and thus $\mu \succ \kappa$. If $\sigma\in \app_{3,3}$, then $\diag(\kappa)\in \app_{3,3}$, which contradicts with $l_1(\kappa)<0$. Thus $\sigma \notin \app_{3,3}$.
This completes the proof.
$\hfill\square$

\vspace{0.2cm}
{\bf Proof of Theorem \ref{nonbou}}

The if part is obvious. We prove the only if part. Suppose $\r$ is a boundary but non-extreme point. We know that $\r$ has full rank and $\r\neq \frac{1}{9}I_9$. Let 
$
0<\epsilon<\frac{1}{18}\min\limits_{j=1,\cdots,8 \atop \l_{j}> \l_{j+1}}\{\l_j-\l_{j+1},\l_9\}.
$
	Using Theorem \ref{thm:main}(ii), there exist two linearly independent diagonal states $\a,\b\in \app_{3,3}$ such that $\r=\frac{1}{2}(\a+\b)$ and $||\r-\a||_2=||\r-\b||_2<\epsilon$. We can write as 
	\begin{eqnarray}
		\label{alpha}
		&&\a=\diag(\a_1,\cdots,\a_9)=\diag(\l_1+x_1,\cdots,\l_9+x_9),\\
		\label{beta}
		&&\b=\diag(\b_1,\cdots,\b_9)=\diag(\l_1-x_1,\cdots,\l_9-x_9),
	\end{eqnarray}
	where $x_i$ are not all zero such that $\sum_{i=1}^9x_i=0$ and $\sqrt{\sum_{i=1}^{9}x_i^2}<\epsilon$.
This implies that $|x_1|,\cdots,|x_9|<\epsilon$. 
	Consequently, for any $1\le j\le 8$ such that $\l_j>\l_{j+1}$, from the range of $\epsilon$, we obtain that
	\begin{eqnarray}
		\label{rrr1}
		&&\a_j-\a_{j+1}=\l_j+x_{j}-\l_{j+1}-x_{j+1}\ge \l_j-\l_{j+1}-2\epsilon>0,\\
		\label{rrr2}
		&&\b_j-\b_{j+1}=\l_j-x_{j}-\l_{j+1}+x_{j+1}\ge \l_j-\l_{j+1}-2\epsilon>0.
	\end{eqnarray}
	Next, for any $1\le k\le 8$ and $l\ge 1$ such that $\l_{k-1}>\l_k=\cdots=\l_{k+l}>\l_{k+l+1}$. 
	Let $\a'=\diag(\a_1',\cdots,\a_9')$ and $\b'=\diag(\b_1',\cdots,\b_9')$, where
	\begin{eqnarray}
		\label{aa''}
		&&\a_i'=\begin{cases} 
			\a_i, & i\neq k,\cdots,k+l \\
			\l_k+\frac{x_k+\cdots+x_{k+l}}{l},&  i=k,\cdots,k+l\\
		\end{cases},\\
		\label{bb''}
		&&\b_i'=\begin{cases} 
			\b_i, & i\neq k,\cdots,k+l \\
			\l_k-\frac{x_k+\cdots+x_{k+l}}{l},&  i=k,\cdots,k+l\\
		\end{cases}.
	\end{eqnarray}
	We have $\r=\frac{1}{2}(\a'+\b')$, where
	\begin{eqnarray}
		\notag
		&&\a'_{k}=\cdots=\a'_{k+l},\\
		\notag
		&&\a'_{k-1}-\a'_{k}=\l_{k-1}+x_{k-1}-\l_{k}-\frac{x_k+\cdots+x_{k+l}}{l}>\l_{k-1}-\l_k-2\epsilon>0,\\
	&&\a'_{k+l}-\a'_{k+l+1}=\l_{k}+\frac{x_k+\cdots+x_{k+l}}{l}-\l_{k+l+1}-x_{k+l+1}>\l_{k+l}-\l_{k+l+1}-2\epsilon>0,
	\end{eqnarray}
	and similarly, $\b'_{k}=\cdots=\b'_{k+l}$, $\b'_{k-1}>\b'_{k}, \b'_{k+l}>\b'_{k+l+1}$. Combining (\ref{rrr1}) with (\ref{rrr2}), we establish that $\a_j'>\a_{j+1}'$ and $\b_j'>\b_{j+1}'$ whenever $\l_j>\l_{j+1}$, at the same time, $\a_k'=\a_{k+1}'$ and $\b_k'=\b_{k+1}'$ whenever $\l_k=\l_{k+1}$. This ensures that the diagonal entries of $\a',\b'$ are in non-increasing order.
	By applying Lemma \ref{hjk} (ii), we have $\a\succ \a'$ and $\b\succ \b'$, and it follows from Theorem \ref{maj} that $\a',\b'\in \app_{3,3}$. It now remains to prove that
	$\a'$ and $\b'$ are linearly independent. If they were linearly dependent, i.e., $\a'=\b'=\r$. From (\ref{aa''}), we would have $x_i=0$ for $i\neq k,\cdots,k+l$ and $x_k+\cdots+x_{k+l}=0$.  
From (\ref{alpha}), we have $\a\succ \r$.  However, recalling that $\r$ is a boundary point, by Theorem \ref{3bb}, $\a\notin \app_{3,3}$, which leads to a contradiction. Hence $\a',\b'$ must be linearly independent. This completes the proof.
$\hfill\square$

\vspace{0.2cm}
{\bf Proof of Theorem \ref{judge}}

Without loss of generality, assume that $\r=\diag(\l_1,\cdots,\l_9)$.
For convenience of the writing, we denote the two matrices $\bma 2t_9&&t_8-t_1&&t_6-t_2\\t_8-t_1&&2t_7&&t_5-t_3\\t_6-t_2&&t_5-t_3&&2t_4\ema$ and $\bma 2t_9&&t_8-t_1&&t_7-t_2\\t_8-t_1&&2t_6&&t_5-t_3\\t_7-t_2&&t_5-t_3&&2t_4\ema$ as $T_1$ and $T_2$, respectively.

(i) From the hypothesis, we know that $L_2(\l)>\mathcal{O}$ and $L_1(\l)$ has deficient rank. If the rank is one, then the first two rows are linearly dependent and hence $\l_1=\l_8+2\sqrt{\l_7\l_9}$. It follows from Lemma \ref{pf3e} (i) that $\l_2=\l_6$, thus $L_1(\l)=\bma 2\lambda_9&&\lambda_8-\lambda_1\\\lambda_8-\lambda_1&&2\lambda_7\ema\oplus 2\l_4$, which cannot have rank one. This is a contradiction. Hence $L_1(\l)$ has rank two, consequently, $D_1$ in (\ref{D1}) can be written as 
$D_1'\oplus 0$, where $D_1'>\mathcal{O}$ has order two.

We begin with the only if part. Suppose $t_1,\cdots,t_9$ is a non-trivial solution of  Eqs.(\ref{fn1})-(\ref{fn2}). We shall show that $\r$ is non-extreme.
Firstly, from (\ref{fn3}), there exists small enough $\epsilon>0$ such that
\begin{eqnarray}
\notag
	&&\l_1+\epsilon t_1\ge \cdots \ge \l_9+\epsilon t_9>0,\\
	\label{zz2}
	&&\l_1-\epsilon t_1\ge \cdots \ge \l_9-\epsilon t_9>0.
\end{eqnarray}
Secondly, from (\ref{fn2}), we have
\begin{eqnarray}
	\label{plm}
	U^T T_1 U=\bma U_1&U_2&U_3\ema^T T_1\bma U_1&U_2&U_3\ema=T_1'\oplus 0,
\end{eqnarray}
where $T_1'\in \cH_2$. Since $D_1'>\mathcal{O}$, 
and $L_2(\l)>\mathcal{O}$, there exists a small enough $0<\delta<\epsilon$ such that
\begin{eqnarray}
	\label{clj1}
	&&D_1'\pm \delta T_1'>\mathcal{O},\\
	\label{clj2}
	&&L_2(\l)\pm \delta T_2>\mathcal{O}.
\end{eqnarray}
Let 
\begin{eqnarray}
	\notag
&&\a=\diag(\l_1+\delta t_1 ,\cdots,\l_9+\delta t_9),\\
&&\b=\diag(\l_1-\delta t_1 ,\cdots,\l_9-\delta t_9).
\end{eqnarray}
We have $\r=\frac{1}{2}(\a+\b)$, and from (\ref{fn1}), $\tr(\a)=\tr(\b)=1$. 
Since $t_i$ are not all zero, $\a$ and $\b$ are linearly independent. 
Further, it follows from (\ref{zz2}) and $\delta<\epsilon$ that the diagonal entries of both $\a,\b$ are in non-increasing order. We have
\begin{eqnarray}
	\label{sccc}
	\notag
	&&L_1(\l(\a))=L_1(\l)+\delta T_1=U((D_1'+\delta T_1')\oplus 0)U^T\ge\mathcal{O},\\
	\notag
	&&L_2(\l(\a))=L_2(\l)+\delta T_2>\mathcal{O},\\
	\notag
	&&L_1(\l(\b))=L_1(\l)-\delta T_1=U((D_1'-\delta T_1')\oplus 0)U^T\ge\mathcal{O},\\
	&&L_2(\l(\b))=L_2(\l)-\delta T_2>\mathcal{O},
\end{eqnarray}
where the inequalities from (\ref{plm}), (\ref{clj1}) and (\ref{clj2}).  
The four inequalities in (\ref{sccc}) jointly imply that $\a,\b\in \app_{3,3}$. So $\r$ is a non-extreme point.

We next prove the if part. Suppose $\r$ is a non-extreme point. We shall show that there exists a non-trivial solution of Eqs.(\ref{fn1})-(\ref{fn2}).
Using Theorem \ref{nonbou}, there exist two linearly independent states $\a=\diag(\a_1,\cdots,\a_9),\b=\diag(\b_1,\cdots,\b_9)\in \app_{3,3}$ such that $\r=\frac{1}{2}(\a+\b)$ where $\a_1\ge\cdots\ge \a_9$ and $\b_1\ge\cdots\ge \b_9$. Let $t_i=\a_i-\l_i=\l_i-\b_i$ for $i=1,\cdots,9$. We know that $t_i$ are not all zero. It suffices to prove that $t_1,\cdots,t_9$ satisfy Eqs.(\ref{fn1})-(\ref{fn2}).

The Eq.(\ref{fn1}) can be verified directly since $\tr(\a)=1+\sum_{i=1}^9t_i=1$. For Eq.(\ref{fn3}), suppose $\l_k=\l_{k+1}$ for some $1\le k\le 8$. Since $\a_k\ge \a_{k+1}$, we have $t_k\ge t_{k+1}$. On the other hand, since $\b_k\ge \b_{k+1}$, we have $t_k\le t_{k+1}$.
Hence $t_k=t_{k+1}$. It remains to prove that $t_1,\cdots,t_9$ satisfy (\ref{fn2}). Since $\a,\b\in \app_{3,3}$, we have
\begin{eqnarray}
	\notag
	&&L_1(\l(\a))=L_1(\l)+T_1\ge \mathcal{O},\\
	&&L_1(\l(\b))=L_1(\l)-T_1\ge \mathcal{O},
\end{eqnarray}
recalling from (\ref{D1}),
\begin{eqnarray}
	\notag
	\label{ASD1}
	&&U^TL_1(\l(\a))U=(D_1'\oplus 0)+U^TT_1U \ge \mathcal{O},\\
	\label{ASD2}
	&&U^TL_1(\l(\b))U=(D_1'\oplus 0)-U^TT_1U \ge \mathcal{O}.
\end{eqnarray}
The two inequalities in (\ref{ASD2}) jointly imply that $U^TT_1U$ must have the form $T_1'\oplus 0$, where $T_1'\in \cH_2$. It follows by a direct calculation that (\ref{fn2}) holds. This completes the proof.

(ii) The hypothesis implies that $L_1(\l)>\mathcal{O}$ and $L_2(\l)$ has deficient rank. One can also verify that the rank of $L_2(\l)$ is two by using Lemma \ref{pf3e} (ii). The remaining proof is totally similar to that of (i).

(iii) The proof is also essentially similar to that of (i). We know that the rank of both $L_1(\l)$ and $L_2(\l)$ is two.
For the only if part, suppose there exists a non-trivial solution $t_1,\cdots,t_9$ of Eqs.(\ref{fn1})-(\ref{kn2}). Then there exists small enough $\epsilon>0$ such that $\a=\diag(\l_1+\epsilon t_1,\cdots,\l_9+\epsilon t_9)>\mathcal{O}$ and $\b=\diag(\l_1-\epsilon t_1,\cdots,\l_9-\epsilon t_9)>\mathcal{O}$, where the diagonal entries are both in non-increasing order. Since $t_1,\cdots,t_9$ satisfy (\ref{fn2}) and (\ref{kn2}), we can restrict $\epsilon$ to be smaller such that $L_1(\l(\a))=L_1(\l)+\epsilon T_1\ge \mathcal{O}$, and  $L_2(\l(\a))=L_2(\l)+\epsilon T_2\ge \mathcal{O}$. Thus $\a\in\app_{3,3}$. Similarly, we have $\b\in\app_{3,3}$. Hence $\r=\frac{1}{2}(\a+\b)$ and $\r$ is a non-extreme point.

For the if part, suppose $\r$ is a non-extreme point. By Theorem \ref{nonbou}, there exist linearly independent states $\a=\diag(\l_1+t_1,\cdots,\a_9+t_9),\b=\diag(\l_1-t_1,\cdots,\a_9-t_9)\in \app_{3,3}$ such that $\r=\frac{1}{2}(\a+\b)$ where the diagonal entries are both in non-increasing order. It is straightforward to see that $t_1,\cdots,t_9$ are not all zero and satisfy Eqs.(\ref{fn1}) and (\ref{fn3}). Further, Eq.(\ref{fn2}) follows from that $L_1(\l)$ has rank two and $L_1(\a)=L_1(\l)+T_1$ and $L_1(\b)=L_1(\l)-T_1$ are simultaneously positive semidefinite. Similarly, Eq.(\ref{kn2}) follows from that $L_2(\l)$ has rank two and $L_2(\l)\pm T_2\ge \mathcal{O}$. Hence $t_1,\cdots,t_9$ is a non-trivial solution of Eqs.(\ref{fn1})-(\ref{kn2}) and the claim holds. 
$\hfill\square$

\vspace{0.2cm}
{\bf Proof of Corollary \ref{ccll}}
	
(i)	We first prove the only if part. If $\r$ is an extreme point, then it is also a boundary point, meaning at least one of $l_1(\l),l_2(\l)$ equals to zero. Since $\r$ has two distinct eigenvalues, through direct computations, we have $\l(\r)$ is the same as one of $\l(\z_k)$ listed in (\ref{twodi}). The claim holds. To prove the only if part, it suffices to prove that the eight states in (\ref{twodi}) are all extreme points of $\app_{3,3}$.
	This can be proved directly by using Theorem \ref{judge}. 
	
(ii) The conditions imply that $l_1(\l)=l_2(\l)=0$. Using Theorem \ref{judge} (iii), to prove $\r$ is an extreme point, it suffices to prove that the following equations have only the trivial solution:
\begin{eqnarray}
	\label{sumj}
	&&t_1+t_2+7t_9=0,\\
	\label{tyu}
	&&\bma 2t_9&&t_9-t_1&&t_9-t_2\\t_9-t_1&&2t_9&&0\\t_9-t_2&&0&&2t_9\ema\cdot U_3=\bma 0\\0\\0\ema.
\end{eqnarray}
Note that we have utilized the fact that Eqs.(\ref{fn2}) and (\ref{kn2}) are identical. Let $U_3=\bma u_{13}&u_{23}&u_{33}\ema^T$. Eq.(\ref{tyu}) can thus be reformulated as
\begin{eqnarray}
	\label{nnre}
	\bma -u_{23}&&-u_{33}&&2u_{13}+u_{23}+u_{33}\\-u_{13}&&0&&u_{13}+2u_{23}\\0&&-u_{13}&&u_{13}+2u_{33}\ema \cdot \bma t_1\\t_2\\t_9\ema=\bma 0\\0\\0\ema,
\end{eqnarray}
where the rank of the coefficient matrix of $t_1,t_2,t_9$ is at least two. On the other hand, we recall from (\ref{D1}) that $(t_1,t_2,t_9)=(a,b,c)$ satisfies (\ref{tyu}) and therefore also satisfies (\ref{nnre}). Hence the solution to (\ref{nnre}) is of the form $(ka,kb,kc)$ for any real number $k$. Combining this with (\ref{sumj}), where $a,b,c$ are all positive, we conclude that $t_1, t_2,t_9$ must be all zero. This completes the proof.

(iii) Suppose $\r$ has at least eight distinct eigenvalues. We aim to prove that $\r$ is a non-extreme point. Since $\r$ is necessarily a boundary point, there are two cases, (A) exactly one of $l_1(\l),l_2(\l)$ equals to zero, (B) $l_1(\l)=l_2(\l)=0$.

(A) Suppose exactly one of $l_1(\l),l_2(\l)$ equals to zero. Here we only prove the case that  $l_1(\l)=0$ and $l_2(\l)>0$. The other case can be proved similarly. Using Theorem \ref{judge} (i), it suffices to prove that Eqs.(\ref{fn1})-(\ref{fn2}) have a non-trivial solution.
Firstly, it is obvious that for Eqs.(\ref{fn1}) and (\ref{fn2}), the rank of the coefficient matrix is at most four. However, 
from Eq.(\ref{fn3}), we deduce that the number of variables involved across
Eqs. (\ref{fn1}) and (\ref{fn2}) is at least eight. Therefore, there must exist a non-trivial solution to Eqs.(\ref{fn1})-(\ref{fn2}). The claim holds.

(B) Suppose $l_1(\l)=l_2(\l)=0$. From Eq.(\ref{fn3}), the number of variables in terms of Eqs. (\ref{fn1}), (\ref{fn2})  and (\ref{kn2}) is at least eight, however, the rank of the coefficient matrix is at most seven. This implies that there must exist a non-trivial solution. Using Theorem \ref{judge} (iii), we obtain that $\r$ is a non-extreme point.
$\hfill\square$

\vspace{0.2cm}
{\bf Proof of Theorem \ref{3ndpd}}

We first prove the if part. Let $\r$ satisfy conditions (I) and (II). Suppose $\rho$ is a non-extreme point. We have
\begin{eqnarray}
	\label{fgg}
	\l_3=\cdots=\l_k\ge \l_{k+1}=\cdots=\l_{3n-5},
\end{eqnarray}
for some $3\le k\le 3n-6$. Using Theorem \ref{thm:main} (ii), there exist two linearly independent diagonal states $\a,\b\in\app_{3,n}$ such that $\r=\frac{1}{2}(\a+\b)$. Define $\mathcal{J}_1:=\{3,\cdots,k\}$ and $\mathcal{J}_2:=\{k+1,\cdots,3n-5\}$. Set
$\cK:=\{1,2,a,b,3n-4,\cdots,3n\}$,
where $a\in J_1$ and $b\in J_2$. So
$\r_{\cK}=\frac{1}{2}(\a_{\cK}+\b_{\cK}).
$
Using Theorem \ref{le:mxn=AS}, we obtain that 
$\r_{\cK}, \a_{\cK}, \b_{\cK}\in \app_{3,3}$. 
Since $\r$ satisfies condition (I), it follows from (\ref{fgg}) and the definition of $\cK$ that $\r_{\cK}$ is an extreme point of $\app_{3,3}$. This implies that $\a_{\cK}$ and $\b_{\cK}$ are linearly dependent. As $a$ and $b$ vary over $\mathcal{J}_1$ and 
$\mathcal{J}_2$ respectively,  the two states $\a_{\cK}$ and $\b_{\cK}$ remain linearly dependent. We conclude that $\a,\b$ are linearly dependent, which contradicts with the assumption. So $\r$ must be an extreme point.

We next prove the only if part.
Suppose $\r$ is an extreme point of $\app_{3,n}$. Firstly, assume condition (II) is violated, specifically, $\l_3>\l_k>\l_{3n-5}$ holds for some $4\le k\le 3n-6$. Let $0<\epsilon<\{\l_3-\l_k,\l_k-\l_{3n-5}\}$. 
The two states $\a:=\frac{1}{1+\epsilon}\diag(\l_1,\cdots,\l_{k-1},\l_k+\epsilon,\l_{k+1},\cdots,\l_{3n})$ and $\b:=\frac{1}{1-\epsilon}\diag(\l_1,\cdots,\l_{k-1},\l_k-\epsilon,\l_{k+1},\cdots,\l_{3n})$ are linearly independent. Moreover, $\a,\b\in \app_{3,n}$, since their three largest and six smallest eigenvalues are proportional to those of $\r$. Hence $\r=\frac{1+\epsilon}{2}\a+\frac{1-\epsilon}{2}\b$ implies that it is a non-extreme point, which contradicts the initial assumption that $\r$ is extreme. Hence $\r$ satisfies condition (II).

Next, suppose condition (II) is satisfied and condition (I) is violated. That is, (\ref{fgg}) holds and $\diag(\l_1,\l_2,\l_3,\l_{3n-5},\cdots,\l_{3n})\in \app_{3,3}$ is a non-extreme point. Using Theorem \ref{nonbou}, there exist two linearly independent (unnormalized) states $\g=\diag(\g_1,\g_2,\g_3,\g_{3n-5},\cdots,\g_{3n})\in \app_{3,3}$  and $\eta=\diag(\eta_1,\eta_2,\eta_3,\eta_{3n-5},\cdots,\eta_{3n})\in \app_{3,3}$ such that $\tr(\g)=\tr(\eta)$ and 
$\diag(\l_1,\l_2,\l_3,\l_{3n-5},\cdots,\l_{3n})=\frac{1}{2}(\g+\eta),
$
where the diagonal entries of $\g,\eta$ are both in non-increasing order.
Let $\g':=\diag(\g_1',\cdots,\g_{3n}')$, $\eta':=\diag(\eta_1',\cdots,\eta_{3n}')$, where 
\begin{eqnarray}
	\notag
	&&\g_i'=\begin{cases} 
		\g_i, & \text{if } i\neq 4,\cdots,3n-6, \\
		\g_3, & \text{if } i=4,\cdots,k,\\
		\g_{3n-5}, & \text{if } i=k+1,\cdots,3n-6
	\end{cases}\\
	\label{ee'}
	&&\eta_i'=\begin{cases} 
		\eta_i, & \text{if } i\neq 4,\cdots,3n-6, \\
		\eta_3, & \text{if } i=4,\cdots,k,\\
		\eta_{3n-5}, & \text{if } i=k+1,\cdots,3n-6
	\end{cases}
\end{eqnarray}
We obtain that $\g',\eta'\in \app_{3,n}$ since the three largest and six smallest eigenvalues of $\gamma',\eta'$ are proportional to those of $\gamma$ and $\eta$, respectively. Further, the two states $\gamma',\eta'$ are linearly independent since $\gamma,\eta$ are linearly independent. Hence $\r=\frac{\tr(\gamma')}{2}\frac{1}{\tr(\g')}\g'+\frac{\tr(\eta')}{2}\frac{1}{\tr(\eta')}\eta'$ implies that it is a non-extreme point. This is a contradiction. Hence condition (I) holds. This completes the proof.
$\hfill\square$

\section{Proof of Theorem \ref{battle}}
\label{james}
(i) Firstly, one can verify from Lemma \ref{le:rankrhoA=m} (i) that  $\frac{3}{mn+2}(\r+\frac{mn-1}{3}\diag(0,\frac{1}{mn-1},\cdots,\frac{1}{mn-1}))\in\as_{m,n}
$. So $\ar(\r)\le \frac{mn-1}{3}$. It remains to prove that $\ar(\r)$ cannot be less than $\frac{mn-1}{3}$. This is equivalent to showing that $\frac{1}{1+t}(\mu+t\r)\notin \as_{m,n}$ for any $m\times n$ state $\mu$ and $t>\frac{3}{mn-1}$.
Assume there exists a state $\mu$ and $t>\frac{3}{mn-1}$ such that $\frac{1}{1+t}(\mu+t\r)\in \as_{m,n}$, where the non-increasing diagonal entries of $\mu$ are denoted as $\mu_1,\cdots,\mu_{mn}$.
We have 
\begin{eqnarray}
	\notag
\frac{1}{1+t}(\mu+t\r)&&\succ \frac{1}{1+t}(\diag(\mu)+t\r)\succ \frac{1}{1+t}\diag(\mu_{mn}+t,\mu_1,\cdots,\mu_{mn-1})\\&&\succ \frac{1}{1+t}\diag(\mu_{mn}+t,\frac{1-\mu_{mn}}{mn-1},\cdots,\frac{1-\mu_{mn}}{mn-1}),
\end{eqnarray}
where the first relation follows from Schur Theorem, the second and third relations follow from Lemma \ref{hjk}(i) and (ii), respectively. By Theorem \ref{maj}, $\frac{1}{1+t}\diag(\mu_{mn}+t,\frac{1-\mu_{mn}}{mn-1},\cdots,\frac{1-\mu_{mn}}{mn-1})\in \as_{m,n}$. However, since $t>\frac{3}{mn-1}$, we have
$\mu_{mn}+t>\mu_{mn}+\frac{3}{mn-1}\ge 3\cdot\frac{1-\mu_{mn}}{mn-1}$.
Therefore $\frac{1}{1+t}\diag(\mu_{mn}+t,\frac{1-\mu_{mn}}{mn-1},\cdots,\frac{1-\mu_{mn}}{mn-1})\notin \app_{m,n}$ as it violates (\ref{eq:e1}). This is a contradiction, hence the assumption is invalid.

Finally, we prove the uniqueness of the optimal state. Suppose
there exists another state $\z$ such that $\frac{3}{mn+2}(\r+\frac{mn-1}{3}\z)\in\as_{m,n}$. By using Schur Theorem and Lemma \ref{hjk} (iii), it can be verified that $\frac{3}{mn+2}(\r+\frac{mn-1}{3}\z)\succ \frac{3}{mn+2}\diag(1,\frac{1}{3},\cdots,\frac{1}{3})$, where the two states have distinct eigenvalues. Recalling from Theorem \ref{maj}, this implies that $\frac{3}{mn+2}\diag(1,\frac{1}{3},\cdots,\frac{1}{3})$ is a non-extreme point. However, this contradicts with Lemma \ref{le:rankrhoA=m} (i). So the assumption does not hold. This completes the proof.

(ii) One can verify that
$\frac{3k}{2n+2k}(\r+\frac{2n-k}{3k}\diag(0,\cdots,0,\frac{1}{2n-k},\cdots,\frac{1}{2n-k}))\in \as_{2,n}$.
Hence $\ar(\r)\le \frac{2n-k}{3k}$. Assume there exists a $2\times n$ state $\a$ and $t>\frac{3k}{2n-k}$ such that $\frac{1}{1+t}(\a+t\r)\in \as_{2,n}$, where the non-increasing diagonal entries of $\a$ are $\a_1,\cdots,\a_{2n}$. Let 
\begin{eqnarray}
\a'=\frac{1}{1+t}\diag(\a_{2n}+\frac{1}{k}t,\cdots,\a_{2n-k+1}+\frac{1}{k}t,\frac{1-\sum_{j=0}^{k-1}\a_{2n-j}}{2n-k},\cdots,\frac{1-\sum_{j=0}^{k-1}\a_{2n-j}}{2n-k}).
\end{eqnarray}
Using Schur Theorem, Lemma \ref{hjk} (i) and (ii), we obtain that $\frac{1}{1+t}(\a+t\r)\succ\a'$, implying $\a'\in\as_{2,n}$. However, as $t>\frac{3k}{2n-k}$, we have
\begin{eqnarray}
	\a_{2n-k+1}+\frac{1}{k}t\ge \a_{2n}+\frac{1}{k}t>\frac{3}{2n-k}\ge 3\cdot \frac{1-\sum_{j=0}^{k-1}\a_{2n-j}}{2n-k}.
\end{eqnarray}
This implies that $\a'\notin\as_{2,n}$, since it violates (\ref{eq:2xnAS}). This contradiction indicates that the assumption is invalid.

We next prove the uniqueness of the optimal state. Suppose there exists another state $\omega$ such that $\frac{3k}{2n+2k}(\r+\frac{2n-k}{3k}\omega)\in \as_{m,n}$. We also have $\frac{3k}{2n+2k}(\r+\frac{2n-k}{3k}\omega)\succ \frac{3k}{2n+2k}\diag(\frac{1}{k},\cdots,\frac{1}{k},\frac{1}{3k},\cdots,\frac{1}{3k})$, where the two states have distinct eigenvalues. This implies that $\frac{3k}{2n+2k}\diag(\frac{1}{k},\cdots,\frac{1}{k},\frac{1}{3k},\cdots,\frac{1}{3k})$ is non-extreme by Theorem \ref{maj}. However, this contradicts Theorem \ref{mmain}. So the assumption does not hold.

(iii) It can be verified that $\frac{n-1}{2-2\sqrt{2}+n}(\r+\frac{3-2\sqrt{2}}{n-1}\diag(0,\cdots,0,\frac{1}{2},\frac{1}{2})
)\in\as_{2,n}$. So $\ar(\r)\le \frac{3-2\sqrt{2}}{n-1}$. 
Assume there exists a state $\b$ and $t>\frac{n-1}{3-2\sqrt{2}}$ such that $\frac{1}{1+t}(\b+t\r)\in \as_{2,n}$, where the non-increasing diagonal entries of $\b$ are $\b_1,\cdots,\b_{2n}$.
Let $
\b'
=\frac{1}{1+t}\diag(\b_{2n-1}+\frac{1}{2}t,\b_{2n}+\frac{1}{2}t,\frac{1-\b_{2n}-\b_{2n-1}}{2n-2},\cdots,\frac{1-\b_{2n}-\b_{2n-1}}{2n-2}).
$
We have $\frac{1}{1+t}(\b+t\r)\succ\b'$, implying $\b'\in\as_{2,n}$. However, since $t>\frac{n-1}{3-2\sqrt{2}}$, we have
\begin{eqnarray}
	\label{ttpp1}
	\b_{2n-1}+\frac{1}{2}t\ge \b_{2n}+\frac{1}{2}t>\frac{1-\b_{2n}-\b_{2n-1}}{2n-2}.
\end{eqnarray}
Further,
\begin{eqnarray}
	\label{ttpp2}
	\notag
	&&\frac{1-\b_{2n}-\b_{2n-1}}{2n-2}+2\sqrt{(\b_{2n}+\frac{1}{2}t)(\frac{1-\b_{2n}-\b_{2n-1}}{2n-2})}\\
\le&&\frac{1-\b_{2n}}{2n-2}+2\sqrt{(\b_{2n}+\frac{1}{2}t)(\frac{1-\b_{2n}}{2n-2})}
	\le\frac{1}{2n-2}+\sqrt{\frac{t}{n-1}},
\end{eqnarray}
where the second inequality is obtained from that the derivative function of  $f(\b_{2n})=\frac{1-\b_{2n}}{2n-2}+2\sqrt{(\b_{2n}+\frac{1}{2}t)(\frac{1-\b_{2n}}{2n-2})}$ is strictly less than zero. Consequently,
\begin{eqnarray}
	\label{ttpp3}
	\b_{2n-1}+\frac{1}{2}t-(\frac{1}{2n-2}+\sqrt{\frac{t}{n-1}})>0,
\end{eqnarray}
where the inequality follows from the fact that the function  $g(t)=\frac{1}{2}t-(\frac{1}{2n-2}+\sqrt{\frac{t}{n-1}})$ is strictly increasing for $t\ge \frac{n-1}{3-2\sqrt{2}}$ and $g(\frac{n-1}{3-2\sqrt{2}})>0$.
Hence (\ref{ttpp1})-(\ref{ttpp3}) jointly imply that $\b'\notin\as_{2,n}$, as it violates (\ref{eq:2xnAS}). This contradiction indicates that the assumption is invalid. The proof of the uniqueness of  the optimal state is similar to that of (ii).

(iv) We first prove the case that $\frac{1}{2}\le a\le \frac{3}{4}$. It can be verified that $\r+(4-4\sqrt{1-a}-2a)\diag(0,0,\frac{1}{2},\frac{1}{2})\in \as_{2,2}$. So $\ar(\r)\le 4-4\sqrt{1-a}-2a$. 
Assume there exists a state $\g$ and $t>\frac{1}{4-4\sqrt{1-a}-2a}$ such that $\frac{1}{1+t}(\g+t\r)\in\as_{2,2}$. Denote the non-increasing ordered diagonal entries of $\g$ as $\g_1,\cdots,\g_4$. Let 
$\g'=\frac{1}{1+t}\diag(\g_4+at,\g_3+(1-a)t,\frac{1-\g_3-\g_4}{2},\frac{1-\g_3-\g_4}{2})$.
We have $\frac{1}{1+t}(\g+t\r)\succ\g'$. Further, since $\frac{1}{2}\le a\le \frac{3}{4}$ and $t>\frac{1}{4-4\sqrt{1-a}-2a}$, one can verify that
\begin{eqnarray}
	\label{jmp}
	at\ge(1-a)t\ge \frac{1}{2}.
\end{eqnarray}
Let $\g''=\frac{1}{1+t}\diag(at,(1-a)t,\frac{1}{2},\frac{1}{2})$.
From Lemma \ref{hjk} (ii)  and (\ref{jmp}), we obtain that $\g'\succ \g''$,   implying $\g''\in\as_{2,2}$. However, since $t>\frac{1}{4-4\sqrt{1-a}-2a}$, we have 
\begin{eqnarray}
	\label{bnn}
at>\frac{1}{2}+2\sqrt{\frac{1}{2}(1-a)t},
\end{eqnarray}
which follows from that the function $h_1(t)=at-\frac{1}{2}-2\sqrt{\frac{1}{2}(1-a)t}$ is strictly increasing and 
$h_1(\frac{1}{4-4\sqrt{1-a}-2a})=0$. Further, (\ref{bnn}) and (\ref{jmp}) jointly imply that $\g''\notin\as_{2,2}$, which is a contradiction. Hence the assumption does not hold. 

We next prove the case that $\frac{3}{4}< a\le 1$. It can be verified that  $\r+(2a-1)\diag(0,\frac{4a-3}{6a-3},\frac{a}{6a-3},\frac{a}{6a-3})\in \as_{2,2}$. So $\ar(\r)\le 2a-1$. 
Assume there exists a state $\eta$ and $t>\frac{1}{2a-1}$ such that $\frac{1}{1+t}(\eta+t\r)\in\as_{2,2}$. Denote the non-increasing diagonal entries of $\eta$ as $\eta_1,\cdots,\eta_4$. Let 
\begin{eqnarray}
	\eta'=\frac{1}{1+t}\diag(\eta_4+at,\eta_3+(1-a)t,\frac{1-\eta_3-\eta_4}{2},\frac{1-\eta_3-\eta_4}{2}).
\end{eqnarray}
We have $\frac{1}{1+t}(\eta+t\r)\succ\eta'$. Further, since $\frac{3}{4}<a\le 1$ and $t>\frac{1}{2a-1}$, one can verify that
$at-(1-a)t>\frac{1}{2}$, consequently,
\begin{eqnarray}
	\label{fnu}
	\eta_4+at\ge at>\max \{\eta_3+(1-a)t,\frac{1}{2}\}.
\end{eqnarray}
Let $\eta''=\frac{1}{1+t}\diag(at,\eta_3+(1-a)t,\frac{1-\eta_3}{2},\frac{1-\eta_3}{2})$. Using (\ref{fnu}) and Lemma \ref{hjk}(ii), one can verify that $\eta'\succ \eta''$. Hence $\eta_2\in\as_{2,2}$. 
Define the function
$h(\eta_3,t):=at-\frac{1-\eta_3}{2}-2\sqrt{\frac{1-\eta_3}{2}(\eta_3+(1-a)t)},
$
for $\eta_3\ge 0$ and $t\ge\frac{1}{2a-1}$. By calculating its partial derivatives, we obtain that $h(\eta_3,t)>h(\eta_3,\frac{1}{2a-1})\ge h(\frac{4a-3}{6a-3},\frac{1}{2a-1})=0$. Combining with (\ref{fnu}), we have $\eta''\notin\as_{2,2}$ as it violates (\ref{eq:2xnAS}). This is a contradiction. So the assumption does not hold. 
This completes the proof. 
$\hfill\square$

\bibliographystyle{unsrt}


\bibliography{ASandAP}

\end{document}